\documentclass[11pt]{article}

\usepackage[T1]{fontenc}
\usepackage[margin=1in]{geometry}
\usepackage{mathtools}
\usepackage{color}
\usepackage{proof}
\usepackage{caption}
\usepackage{bbm}
\usepackage{amssymb}
\usepackage{microtype}
\usepackage{xcolor}
\usepackage{enumitem}
\usepackage{amsthm}
\usepackage{thmtools,thm-restate}
\usepackage{mdframed}
\usepackage{authblk}

\usepackage{amsmath}
\usepackage{amsfonts}
\usepackage{url}
\usepackage{qcircuit}

\usepackage[colorlinks,citecolor=blue]{hyperref}
\newtheorem{theorem}{Theorem}[section]
\newtheorem{lemma}[theorem]{Lemma}
\newtheorem{definition}[theorem]{Definition}

\newtheorem{corollary}[theorem]{Corollary}
\newtheorem{claim}[theorem]{Claim}
\newtheorem{remark}[theorem]{Remark}

\newcounter{prob}
\newtheorem{problem}[prob]{Problem}

\newcommand{\eq}[1]{\hyperref[eq:#1]{(\ref*{eq:#1})}}
\renewcommand{\sec}[1]{\hyperref[sec:#1]{Section~\ref*{sec:#1}}}
\newcommand{\thm}[1]{\hyperref[thm:#1]{Theorem~\ref*{thm:#1}}}
\newcommand{\lem}[1]{\hyperref[lem:#1]{Lemma~\ref*{lem:#1}}}
\newcommand{\cor}[1]{\hyperref[cor:#1]{Corollary~\ref*{cor:#1}}}
\newcommand{\itm}[1]{\hyperref[itm:#1]{\ref*{itm:#1}}}
\newcommand{\app}[1]{\hyperref[app:#1]{Appendix~\ref*{app:#1}}}
\newcommand{\dfn}[1]{\hyperref[dfn:#1]{Definition~\ref*{dfn:#1}}}
\newcommand{\fig}[1]{\hyperref[fig:#1]{Figure~\ref*{fig:#1}}}
\newcommand{\clm}[1]{\hyperref[clm:#1]{Claim~\ref*{clm:#1}}}
\newcommand{\stp}[1]{\hyperref[stp:#1]{Step~\ref*{stp:#1}}}
\newcommand{\prot}[1]{\hyperref[prot:#1]{Protocol~\ref*{prot:#1}}}
\newcommand{\prob}[1]{\hyperref[prob:#1]{Problem~\ref*{prob:#1}}}
\newcommand{\rmk}[1]{\hyperref[rmk:#1]{Remark~\ref*{rmk:#1}}}
\newcommand{\tabl}[1]{\hyperref[tabl:#1]{Table~\ref*{tabl:#1}}}

\newcommand{\bra}[1]{\langle #1 \vert}
\newcommand{\ket}[1]{\vert #1 \rangle}
\newcommand{\proj}[1]{\vert #1\rangle\!\langle #1\vert}

\newcommand{\tr}[0]{\mathrm{tr}}

\let\originalleft\left
\let\originalright\right
\renewcommand{\left}{\mathopen{}\mathclose\bgroup\originalleft}
\renewcommand{\right}{\aftergroup\egroup\originalright}

\newcommand{\A}[0]{\mathcal{A}}
\newcommand{\B}[0]{\mathcal{B}}
\newcommand{\C}[0]{\mathcal{C}}
\newcommand{\D}[0]{\mathcal{D}}
\newcommand{\E}[0]{\mathcal{E}}
\newcommand{\F}[0]{\mathcal{F}}

\newcommand{\I}[0]{\mathcal{I}}
\newcommand{\K}[0]{\mathcal{K}}

\renewcommand{\O}[0]{\mathcal{O}}
\renewcommand{\P}[0]{\mathcal{P}}
\newcommand{\R}[0]{\mathcal{R}}

\newcommand{\U}[0]{\mathcal{U}}

\newcommand{\X}[0]{\mathcal{X}}
\newcommand{\Y}[0]{\mathcal{Y}}

\DeclareMathOperator*{\Exp}{\mathbb{E}}
\DeclareMathOperator{\poly}{poly}
\DeclareMathOperator{\negl}{negl}
\newcommand{\class}[1]{\mathsf{#1}}

\usepackage{xspace}

\newcommand{\prover}[0]{P}
\newcommand{\verifier}[0]{V}

\newcommand{\bit}{\{0,1\}}

\newcommand{\Id}[0]{\mathbbm{1}}
\newcommand{\BQP}[0]{\mathsf{BQP}}

\newcommand{\BPP}[0]{\mathsf{BPP}}
\newcommand{\BQNC}[0]{\mathsf{BQNC}}
\newcommand{\QNC}[0]{\mathsf{QNC}}
\newcommand{\dSSP}[0]{d\text{-}\mathsf{SSP}}
\newcommand{\xtest}[0]{\mathsf{X}\text{-}\mathsf{TEST}}
\newcommand{\ztest}[0]{\mathsf{Z}\text{-}\mathsf{TEST}}
\newcommand{\rigid}[0]{\mathsf{RIGID}}
\newcommand{\comp}[0]{\mathsf{COMP}}
\newcommand{\query}[0]{\mathsf{QUERY}}

\newcommand{\Gen}[0]{\textsc{Gen}}
\newcommand{\Chk}[0]{\textsc{Chk}}
\newcommand{\Inv}[0]{\textsc{Inv}}
\newcommand{\Samp}[0]{\textsc{Samp}}
\newcommand{\Supp}[0]{\textsc{Supp}}
\newcommand{\CVQD}[0]{\mathsf{CVQD}}

\newcommand{\gatew}[1]{*+<.6em>{#1} \POS ="i","i"+UR;"i"+UL **\dir{-};"i"+DL **\dir{-};"i"+DR **\dir{-};"i"+UR **\dir{-},"i" }

\newmdtheoremenv[backgroundcolor=gray!5,
                 linewidth=0pt,
                 innerleftmargin=4pt,
                 innerrightmargin=4pt,
                 innertopmargin=-2pt,
                 innerbottommargin=4pt,
                 splitbottomskip=4pt]{protocol}[prot]{Protocol}
\let\oldprotocol\protocol
\renewcommand{\protocol}{\oldprotocol\normalfont}

\newlist{proto}{description}{1}
\setlist[proto]{align=right,labelindent=1em,labelwidth=1.5cm,leftmargin=!,itemsep=0pt}

\begin{document}
\title{Classical verification of quantum depth}
\author[1]{Nai-Hui Chia}
\author[2]{Shih-Han~Hung}
\date{}

\affil[1]{Department of Computer Science, Indiana University Bloomington}
\affil[2]{Department of Computer Science, University of Texas at Austin}

\maketitle
\begin{abstract}

    We present two protocols for classical verification of quantum depth. Our protocols allow a purely classical verifier to distinguish devices with different quantum circuit depths even in the presence of classical computation. We show that a device with quantum circuit depth at most $d$ will be rejected by the verifier even if the prover applies additional polynomial-time classical computation to cheat. On the other hand, the verifier accepts a device which has quantum circuit depth $d'$ for some $d'>d$.
    In our first protocol, we introduce an additional untrusted quantum machine which shares entanglements with the target machine. 
    Applying a robust self-test, our first protocol certifies the depth of the target machine with  information theoretic security and nearly optimal separation. 
    The protocol relies on the oracle separation problem for quantum depth by Chia, Chung and Lai [STOC 2020] and a transformation from an oracle separation problem to a two-player non-local game.
    Our second protocol certifies the quantum depth of a single device based on quantum hardness of learning with errors. 
    The protocol relies on the noisy trapdoor claw-free function family and the idea of pointer chasing to force the prover to keep quantum coherence until all preceding message exchanges are completed.  
    To our knowledge, we give the first constructions for distinguishing hybrid quantum-classical computers with different circuit depths in unrelativized models.

\end{abstract}

\section{Introduction}
Quantum circuit depth is an essential consideration when evaluating the power of near-term quantum devices.
Quantum computers with many qubits have been recently implemented~\cite{IBM,GOOGLE,IonQ,Rigetti}; however, these computers have limited quantum circuit depth due to the noisy gates and short coherence time. Hence, how to leverage the power of these small-depth quantum devices becomes a practical challenge as well as a fascinating question in quantum complexity theory.

Indeed, Aaronson and Chen showed that small-depth quantum computers can demonstrate so-called ``Quantum Supremacy''~\cite{AC17} on the random circuit sampling problem, which means that quantum computers can efficiently solve the problem that is intractable for classical machines. Arute et al.~\cite{Arute2019} reported the results of experiments on demonstrating quantum supremacy by using superconducting quantum computers of Google.\footnote{New classical algorithms are found for solving the problem in few days (by estimation)~\cite{IBM_supremacy}, which implies that random circuit sampling of the size in the experiment in~\cite{Arute2019} might not be classically intractable. However, even these new classical algorithms are slower than the quantum one (that solves the problem in $200$ seconds); therefore, the experiments showed quantum advantages on the problem.} 
In the near term, the coherence time seriously limits the usable lifespan of quantum states. Thus information processing with a small and noisy quantum device has become a central topic in field of quantum computing.

Among the computational models that use small quantum devices, hybrid quantum-classical computing that interleaves classical computers with quantum devices is a natural approach to use the power of small-depth quantum circuits. This hybrid approach has been gaining much attention recently and might be able to surpass the capabilities of classical machines on some real-world problems, such as molecular simulation~\cite{MRBG16}, optimization problems~\cite{FGG14}, etc. 
Notably, Cleve and Watrous~\cite{CW00} proved that the quantum Fourier transform can be implemented in logarithmic quantum depth in this model. %
This implies that quantum algorithms for Abelian hidden subgroup problems, such as Shor's factoring algorithm, can be implemented in logarithmic quantum circuit depth. 

The results above indicate that quantum computers with circuit depth beyond certain thresholds are able to demonstrate quantum advantages. Seeing the possible applications of small-depth quantum devices, one might start wondering: %
\begin{itemize}
    \item[] \textsl{Can we certify if a computer has sufficient quantum depth for quantum advantages? } 
\end{itemize}

An answer to the question is to find some problem, give an efficient algorithm that only requires small-depth quantum circuits, and prove that no algorithm using strictly smaller quantum depth achieves the same time complexity. For instance, the aforementioned results~\cite{AC17,CW00} showed separations between small-depth quantum circuits and classical computers under plausible computational assumptions. 
That is, based on the assumption that a problem is classically hard, a device which can solve some problem in a reasonable time frame must exhibit quantum power.

Another possible approach is designing cryptographic protocols that demonstrate the \emph{quantumness} of a quantum device%
~\cite{Brakerski18,brakerski2020simpler,hirahara2021test,liu2021depth}. %
In these protocols, the classical verifier sends the description of a cryptogrpahic hash function $f$ and random coins to challenge the prover to answer information about $f$. 
It is guaranteed that only a prover which performs quantum computation will successfully answer these challenges with high probability. 
While these protocols seem to be satisfying proposals for demonstrating quantumness, there is a caveat: for a quantum prover to succeed, it is required to evaluate $f$ coherently, and thus the implementation of $f$ with quantum gates sets a lower bound on the resource requirement.
To address the issue, in subsquent works, Hirahara and Le Gall~\cite{hirahara2021test} and Liu and Gheorghiu~\cite{liu2021depth} independently showed that these protocols only requires a hybrid computation that only uses a constant-depth quantum circuit using different approaches.

It is worth noting that these two approaches, in state of the art, mainly focus on distinguishing quantum computers from classical ones. 
They do not directly translate into ones that separate quantum computers with different quantum resources. It is unclear if we can show that these protocols or problems cannot be tackled using smaller quantum depth. %

In this work, we give ``fine-grained'' solutions to the question in the following scenario: Suppose Bob claims that he has a computer with quantum circuit depth larger than $d$. 
Can Alice, who only has a classical machine, catch Bob is cheating if Bob only has quantum circuit depth at most $d$? 
Of course, Alice shall assume that Bob might use a powerful classical machine to cheat. Here, we are actually asking for protocols that allow a classical verifier to verify if a prover has a quantum circuit with depth at least $d$ in the presence of the prover's powerful classical machine. We call such protocols \emph{Classical Verification of Quantum Depth} (CVQD). %
\vspace{-3mm}
\paragraph{The problem for separating quantum depth.} Chia, Chung, and Lai~\cite{CCL19} introduced the $d$-Shuffling Simon's Problem Problem ($d$-SSP) that separates $d$- from $(2d+1)$-depth quantum circuits in the presence of polynomial-time classical computation. One straightforward approach is to use this problem to certify quantum depth by asking the computer to solve $d$-SSP. However, this approach does not lead to a solution to classical verification of quantum depth since $d$-SSP is an oracle problem that requires quantum access to the oracle for efficient quantum algorithms. Therefore, we need new ideas for our purpose.

\vspace{-3mm}
\subsection{Main results}

In this work, we give an affirmative answer to the question by showing two $\CVQD$ protocols that can distinguish quantum circuits with different depth in the presence of polynomial-time classical computation. We first give definitions of the two $\CVQD$ protocols that we consider in this work. 

We consider the setting where a single quantum machine is being tested by a classical verifier.
The verifier should reject if the quantum depth no more than $d$, and accept if the quantum depth is at least $d'>d$.
The machines are promised to be in one of the cases.
\begin{definition}[$\CVQD(d,d')$, informal]
\label{dfn:cvqd_1p}
Let $d,d'\in \mathbb{N}$ and $d'>d$. Let $\prover_A$ be a bounded-depth quantum circuit with classical polynomial-time computation. %
Let $V$ be a classical verifier. A $\CVQD(d,d')$ protocol that separates quantum circuit depth $d$ from $d'$ satisfies the following properties:
\begin{itemize}
    \item \textup{\bf Completeness:} If $P_A$ has quantum circuit depth at least $d'$, then $\langle V,P_A\rangle$ accepts with probability at least $2/3$. 
    \item \textup{\bf Soundness:} If $P_A$ has quantum circuit depth at most $d$, then for any polynomial-time $P_A$, $\langle V,P_A\rangle$ accepts with probability at most $1/3$. 
\end{itemize}
\end{definition}

We also consider protocols that consist of two provers which are not allowed to communicate with each other, but may share entanglements. 
In this setting, one prover $\prover_A$ is the target machine being tested.
We add another prover $\prover_O$ to help certify the quantum depth, but neither of the provers is trusted by the classical verifier.

\begin{definition}[$\CVQD_2(d,d')$, informal]
\label{dfn:cvqd_2p}
Let $d,d'\in \mathbb{N}$ and $d'>d$. Let $P_A$ be a bounded-depth quantum circuit with classical polynomial-time computation. %
Let $P_O$ be an unbounded quantum prover and $V$ be a classical verifier. A $\CVQD_2(d,d')$ protocol that separates quantum circuit depth $d$ from $d'$ satisfies the following properties: %
\begin{itemize}
    \item \textup{\bf Non-locality:} $P_O$ and $P_A$ share arbitrarily many EPR pairs and are not allowed to communicate with each other once the protocol starts. 
    \item \textup{\bf Completeness:} If $P_A$ has quantum circuit depth at least $d'$, then there exists $P_O$ and $P_A$ such that $\langle V,P_O,P_A\rangle$ accepts with probability at least $2/3$.
    \item \textup{\bf Soundness:} If $P_A$ has quantum circuit depth at most $d$, then for any $P_O$ and polynomial-time $P_A$, $\langle V,P_O,P_A\rangle$ accepts with probability at most $1/3$. 
\end{itemize}
\end{definition}

In both definitions, the verifier accepts if $\prover_A$ has a quantum circuit with depth at least $d'$,
and rejects any dishonest prover which might interleave its small-depth quantum circuit (depth at most $d$) with a polynomial-time classical algorithm. %
As defined by Chia, Chung and Lai \cite{CCL19}, there are two schemes, where the quantum process interleaves a quantum machine with a classical one, called $d$-CQ and $d$-QC schemes.
Briefly, the $d$-CQ scheme allows a classical algorithm to query a $d$-depth quantum circuit polynomially many times, and the $d$-QC scheme allows a $d$-depth quantum circuit to access polynomial-time classical algorithms after each layer of $1$-depth circuit. We aim to design protocols to rule out cheating provers using both schemes. 

We show there exist constructions of $\CVQD$ and $\CVQD_2$. In particular, we show the following result. 
\begin{theorem}[Informal]
\label{thm:informal_1}
Let $d\in \mathbb{N}$.
\begin{enumerate}
    \item There exists a two-prover $\CVQD_2(d,d+3)$ protocol $\langle V,P_A,P_O\rangle$ that is unconditionally secure with inefficient honest $\prover_O$ and $V$. 
      Moreover, honest $\prover_O$ and $V$ can be efficient assuming the existence of quantum-secure pseudorandom permutation (qPRP). 
    \item For polynomially bounded function $d$ and constant $d_f$, there exists a $\CVQD(d,d+d_f)$ protocol under the \textsl{QLWE} assumption.
\end{enumerate}
\end{theorem}
Here, the \textsl{QLWE} assumption assumes that the Learning With Error (LWE) problems in hard for any quantum polynomial-time algorithms.\footnote{In fact, it is sufficient to assume that QLWE is hard for a $d$-depth hybrid machines.}
The constant $d_f$ is the quantum circuit depth for implementing a particular function.

\paragraph{Comparing the two results in \thm{informal_1}.} The two results in \thm{informal_1} are incomparable. The second result achieves single-prover $\CVQD(d,d+d_f)$ under QLWE. It is worth noting that we only know that $\prover_A$ that implements $(d+d_f)$-QC schemes can be accepted in this protocol; in contrast, $d'$-CQ schemes might require $d'$ to be larger than $d+d_f$ to be accepted. On the other hand, although the first result (a construction of $\CVQD_2(d,d+3)$) requires an additional (and untrusted) quantum prover, it has the following advantages: 
(1) its separation is nearly optimal ($d$ versus $d+3$), 
(2) it achieves information theoretic security, 
(3) $\prover_O$  and the verification can be made efficient by only assuming the existence of qPRPs in a query model, and 
(4) $\prover_A$ that implements either $(d+3)$-CQ or $(d+3)$-QC schemes can be accepted. 

To prove the first result in \thm{informal_1}, we provide a framework that transforms a quantum oracle separation into a two-prover protocol. 
\begin{theorem}[Informal]
\label{thm:informal_framework}
Let $\mathcal{C}$ and $\mathcal{C}'$ be two complexity classes. Let $L^{O}$ be an oracle problem such that $L^{O}\in \mathcal{C}^O$ and $L^{O}\notin \mathcal{C}'^O$. Then, there exists a two-prover protocol $\langle V,P_A,P_O\rangle$ two real numbers $c,s\in[0,1]$ satisfying $c-s=1/\poly(n)$ for size $n$ of input such that the following conditions hold.
\begin{itemize}
    \item \textup{\bf Completeness:} If $P_A$ can solve problems in $\mathcal{C}$, then there exists $P_O$ such that $\langle V,P_A,P_O\rangle$ accepts with probability at least $c$.
    \item \textup{\bf Soundness:} If $P_A$ can only solve problems in $\mathcal{C}'$, then for any $P_O$, $\langle V,P_A,P_O\rangle$ accepts with probability at most $s$. 
    \item \textup{\bf Classical verification:} $V$ is classical, and the runtimes of $V$ and the honest $P_O$ depend on the number of queries for solving $L^{O}$ and the complexity for implementing $O$.
\end{itemize}
\end{theorem}

We can transform $\dSSP$ (an quantum oracle problem by Chia, Chung and Lai~\cite{CCL19} for separating quantum depth) into a two-prover $\CVQD$ protocol under the framework in \thm{informal_framework}. However, the separation is not as good as in \thm{informal_1}. We further show that by modifying the original $\dSSP$, the separation can be improved to $d$ versus $d+3$. We call the new problem ``in-place $\dSSP$.'' %
Then, we transform in-place $\dSSP$ into a two-prover $\CVQD(d,d+3)$ protocol following \thm{informal_framework}.

\subsection{Technical overview}

In this section, we give a brief overview for proving the main theorems. 

\subsubsection{A two-prover protocol with optimal depth separation}
\paragraph{The problem $d$-SSP for separating quantum depth.} 
Our protocol can be seen as a two-player instantiation of the algorithms for solving $d$-SSP in~\cite{CCL19}, an oracle problem that distinguishes $d$- from $(2d+1)$-depth quantum circuits. 

The problem is a shuffled version of the Simon's problem.
Recall that for the ``plain'' Simon's problem, a constant-depth algorithm is sufficient to output the hidden shift.
To turn the problem into one that certifies large quantum query depth, Chia, Chung and Lai proposed the $d$-Shuffling Simon's problem ($\dSSP$) \cite{CCL19}.
The algorithm is given oracle access to $d+1$ functions $f_0,\ldots,f_d$, where $f_1,\dots,f_{d-1}$ are random permutations on an exponentially larger set, and the last function $f_d$ is a $2$-to-$1$ function such that $f_d\circ\cdots\circ f_2\circ f_1(x) = f(x)$ for a Simon's function $f$.
We call the functions $f_0,\ldots,f_d$ to be a $d$-shuffling of a Simon's function $f$.
The task is to find the hidden shift. 

It is obvious that $\dSSP$ remains easy for a $(2d+1)$-depth quantum algorithm which simulates a query to $f$ using two queries to each function in $\{f_0,\ldots,f_{d-1}\}$ and one query to $f_d$: 
first query $f_0,\ldots,f_d$ in sequence to get 
\begin{align}
    \ket{x,y}\mapsto\ket{x,f_0(x),f_1\circ f_0(x),\ldots,f_{d-1}\circ \cdots\circ f_0(x), f(x)\oplus y},
\end{align}
and then query the first $d$ functions in the reverse order to reset the intermediate registers to back zero states. On the other hand, any polynomial-time algorithm with quantum depth at most $d$ cannot solve $\dSSP$. This follows from the intuitions that one needs to make $(d+1)$-sequential quantum queries to $f_0,\dots,f_d$ in order for evaluating $f$ on a uniform superposition, %
and only $f_d$ in an exponentially small random subset of the domain has information about $f$. Thus, any polynomial-time algorithm without sufficient quantum depth cannot even evaluate $f$ in superposition.

To turn the problem into a protocol that certifies quantum depth, an idea is to have the verifier $V$ play the role of the oracle, and checks if prover $P_A$ outputs the hidden shift.
The resulting protocol is quite straightforward: the prover $P_A$ is allowed to perform arbitrary quantum computation (subject to its quantum resources) between message exchanges with the verifier.
In the intermediate rounds, $V$ computes the quantum circuits of the oracles on the state given by $P_A$, and sends the resulting quantum state back.
At the end, the verifier accepts if $P_A$ outputs the hidden shift.
The analysis of the protocol is also straightforward.
As long as the verifier implements the quantum-accessible oracles $f_0,\ldots,f_d$ reliably between the computation performed by the prover, the completeness and soundness directly follows from the result of Chia, Chung and Lai \cite{CCL19}.

However, this approach has two drawbacks: First, the verifier needs to reliably implement a large QRAM to support quantum access to the oracle.
This requires a reliable large-scale quantum computer that can solve problems in quantum polynomial time.
Moreover, it requires reliable quantum communication between the prover and the verifier.
None of the requirements seems to be within the reach in the near future.

In this paper, we give constructions that allows a purely classical veriifer to certify quantum depth. 
Our first protocol is to rely on the technique of self-testing to certify the untrusted quantum servers sharing entanglements.
In particular, we apply a sequence of transformations from the aforementioned straightforward approach into one that has a weak requirement on the verifier, i.e., it runs in probabilistic polynomial time.
We briefly introduce the techniques as follows.

\vspace{-3mm}
\paragraph{Delegating the oracle to another quantum prover.} To achieve purely classical verification, we introduce another untrusted prover, denoted $P_O$, which may share entanglements with $P_A$ but they are not allowed to communicate with each other once the protocol starts. 
The verifier delegates the oracle computations to $P_O$, and checks if $P_A$ outputs the hidden shift in the end.
To make ``queries,'' $P_A$ forwards a quantum state by quantum teleportation.

To ensure that $P_O$ behaves honestly, we modify the EPR protocol by Broadbent~\cite{Bro18} to verify the computation of $P_O$.
To understand how this works, let us recall some idea of the protocol.
The original Broadbent protocol allows a weakly quantum verifier to delegate a quantum computation to the prover. To show that the prover has to be honest, the computation is made indistinguishable from two tests ($X$-test and $Z$-test). 
These tests are used to check if the prover's attack is trivial on the single bit the verifier aims to learn from the prover.

However, to apply the protocol to our problem, there are two caveats that remain to solve. 
First, the Broadbent protocol is designed for verifying a $\BQP$-complete language. 
An instance in the language is a classical description of a unitary $U$ with the promise that sampling the first qubit of $U\ket{0}$ by performing a standard basis measurement yields a 0 with probability at least 2/3, or at most 1/3.
In our setting, we do not have such a promise.
Secondly, the protocol only guarantees that the output $b$ is an encryption of the random variable close to sampling $U\ket{0}$ by performing a standard basis measurement on the first qubit, provided the prover passes the tests with high probability.
For our purposes, we would need to show if $P_O$'s output state $\rho_i$ on each query $\ket{\psi_i}$ is close to $O\ket{\psi_i}$ for each query $i\in\{0,1,\ldots,d\}$ in a reasonable metric.

We show that with a modification, our variant of the Broadbent protocol is rigid in the sense that every prover that is accepted in our variant with probability $1-\epsilon$ must output a state $\rho_i$ which is $O(\epsilon)$-close the ideal state $O\ket{\psi_i}$ in trace distance.
The modification requires a quantum channel which allows a transmission of $\poly(n)$ qubits between $P_O$ and $V$, but the requirement is not necessary when we turn the protocol into a purely classical verification.

\vspace{-3mm}
\paragraph{Dequantizing the verification.} 
We further dequantize the quantum verification and communication by applying the Verifier-on-a-Leash protocol (also called the Leash protocol) by Coladangelo, Grilo, Jeffery and Vidick~\cite{CGJV19}.  %
In a high level, the idea is to add another prover to perform the measurements by the quantum verifier in the Broadbent protocol, and check if the added prover behaves honestly.

To transform an oracle separation into purely classical verification, one possible approach is to add a third prover $P_V$ to help certify that $P_O$ behaves as intended. 
More concretely, the classical verifier asks $P_A$ to perform the computation between queries, and $P_O$ to apply the quantum circuit of the oracle. 
To check $P_O$ behaves as intended, a third player $P_V$ is added to perform the measurements in the bases determined from the rules of the Broadbent protocol. In our settings, none of the provers are assumed to be trusted. Thus it is necessary to verify $P_V$ performs the measurements in the correct bases. Thus the verifier challenges $P_V$ and $P_O$ to run either the protocol for verifying $P_O$ or a rigidity test to certify $P_V$, and the two choices are made indistinguishable to $P_V$'s viewpoint. 

However, this approach does not work directly. 
More specifically, for the security to hold, it is crucial that $P_O$ does not distinguish the computational round and the test rounds. %
In this approach, $P_O$ interacts with $P_A$ via quantum teleportation to implement the original query algorithms in the computation round, whereas to certify $P_O$'s behavior, the classical verifier must ask $P_O$ to interact with $P_V$ in the test round. 
Hence, $P_O$ can determine the round type and cheat. 
Moreover, another drawback with this approach is that it requires three provers.

We can fix the issues about the aforementioned three-prover protocol by \emph{asking $P_A$ to play the role of $P_V$ simultaneously}. 
To explain how this works, we consider the following protocol for a single-query algorithm:
Initially, the verifier chooses to run the computation, X-test, Z-test or rigidity test.
The prover $P_A$ prepares an (arbitrary) initial $n$-bit quantum state $\ket{\psi_0}$ and teleports three states $\ket{\psi_0},\ket{0^n},\ket{+^n}$ to disjoint random subsets of $P_O$'s halves of EPR pairs (the other halves are held by $P_A$).
Note that the subsets are chosen by the verifier, but it does not reveal the underlying states.
If any of the first three tests is chosen, the prover $P_O$ performs the computation $O$ on one of the subsets specified by the verifier. 
Note that since these three states are encrypted with quantum one-time pad, $P_O$ cannot distinguish them and thus the round type. 
To perform the computation $O$, a set $S$ of EPR pairs shared by $P_O$ and $P_A$ is used to implement gadgets for computing $O$. To be more specific, the verifier asks $P_A$ to perform measurements on $S$ in random bases, and chooses a subset (of $S$) on which $P_A$ is specified to perform measurements in desirable bases determined by the rules of the Broadbent protocol according to the round type. The verifier then tells $P_O$ to use the subset in $S$ to compute $O$. The random-basis measurement on $S$ is to certify the behavior of $P_A$ in the rigidity test. Roughly speaking, when the rigidity test is chosen, the verifier can check (with help of $P_O$) if $P_A$ performs the measurements on the EPR pairs in the random bases chosen by the verifier. Since $P_A$'s behavior for all four tests are measurements in random bases, whether a rigidity test is executed is unknown to $P_A$. Note that although $P_O$ can learn if a rigidity test is executed, it does not affect the security since $P_O$ has no chance to reveal this to $P_A$. 

However, there are a couple of issues that remain to address when considering multiple rounds of interaction between $P_O$ and $P_A$.
First, some tests running on more than one query can potentially reveal the type of the test.
More specifically, if the verifier chooses to run the rigidity test, then $P_O$ would certainly learn an application of the oracle unitary $O$ is not necessary for this round.
The prover $P_A$ can possibly detect the choice of the test by observing the input state and the resulting state using, say, a swap test. 
Furthermore, to reflect the actual performance of the query algorithm, it is crucial that with sufficiently large probability, no test has been applied throughout the protocol. 
This is because when the computation is not performed in this round, computation applied in the following rounds will not yield a useful result (e.g., outputting the hidden shift for $\dSSP$), even when the provers opt to follow the protocol honestly.
If the tests are nevertheless executed with very low probability, the provers may deviate from the protocol seriously.

We show that it suffices that the verifier randomly selects to certify one random query and trusts all the other queries, and with probability $\Theta(1/q)$, no test is executed for a $q$-query protocol. 
For the selected query, the verifier either asks $P_O$ to certify $P_A$'s measurements, or $P_A$ to certify $P_O$ performs the oracle unitary $O$ by running the test phases in our variant of the Broadbent protocol.
If the provers pass the test, the verifier accepts and terminates the protocol.
Since the knowledge of the round type for certifying a query can only lead to an attack on the following queries, verification on a random query can prevent these issues from breaking the soundness of the protocol.

\vspace{-3mm}
\paragraph{Putting things together.}
We then combine our aforementioned tests to turn an oracle separation problem into a two-player protocol.
In particular, we show that with a suitable choice of the weights of entering each test, the completeness-soundness gap shrinks by at most an inverse polynomial multiplicative factor in the number of queries.
More formally, we prove the implication by reduction. Suppose that in the protocol, there are provers $P_A,P_O$ such that 
$P_A$ is subject to its quantum resources and they break the soundness.
Then we construct a query algorithm which succeeds with sufficiently large probability to break the soundness guarantee in the associated relativized world.
Since the oracle separation problem distinguishes the quantum complexity classes,
the resulting protocol yields a completeness-soundness gap $1/\poly(q)$ for a $q$-query algorithm. 

Given that $\dSSP$ is a oracle separation problem between a hybrid $d$-depth and a hybrid $(2d+1)$-depth computation, we conclude that our transformation yields a construction of $\CVQD_2(d,2d+1)$ with gap $1/\poly(d)$.
We apply a sequential repetition to amplify the gap to constant.
The repetition itself does not require an increase of the quantum depth of $P_A$ since the same hybrid computation can be reused. %

\vspace{-3mm}
\paragraph{Efficient instantiation.}
We have shown that an oracle separation problem implies a two-player protocol that distinguishes hybrid quantum computation with different quantum depth. 
However, to succeed in the protocol honestly, $V$ must sample an oracle from a distribution $\D$ which is not known to be efficiently samplable, and $P_O$ must perform a quantum circuit that implements $O$.
In the problem $\dSSP$, the oracles consist of random permutations.
By a counting argument, most of the permutations does not have an efficient implementation.

To address the issue, we leverage oracle indistinguishability in the associated relativized world.
More concretely, suppose that for the distribution $\D$ of random $d$-shuffling of a random Simon's function, there is an efficiently samplable distribution $\D'$ which is indistinguishable from $\D$. Then in the two-player protocol, when the efficient verifier samples the oracle according to $\D'$, the soundness error is increased negligibly.
The idea for showing this directly follows from our proof for showing an oracle separation implies a two-player protocol.
For every query algorithm $\A$ that has small quantum depth, it succeeds with probability at most $p$ when the oracle is sampled from $\D$. Then replacing $\D$ with $\D'$, by the oracle indistinguishability, $\A$ succeeds with probability at most negligibly close to $p$.
Applying the transformation with $\D'$ yields a sound two-player protocol with efficient $P_O$ and $V$.

We show how to give a distribution $\D'$ using quantum-secure pseudorandom permutations (qPRP) against adversaries making queries to the permutation and its inverse.
In particular, to sample a pseudorandom $d$-shuffling of a Simon's function, first sample $d$ independent keys $k_0,\ldots,k_{d-1}$ from the key space of the pseudorandom permutation $P$ and let $f_i=P(k_i,\cdot)$ for $i\in\{0,1,\ldots,d-1\}$.
For the last function, again by a counting argument, not every Simon's function has an efficient implementation.
We observe %
that every Simon's function can be computed by composing a permutation and an efficiently computable function that is constant on every one dimensional affine subspace of the form $\{x,x\oplus s\}$.
This implies that sampling a random Simon's function can be done by sampling a random shift and a random permutation.
Replacing the random permutation with a pseudorandom permutation yields a pseudorandom Simon's function.

\vspace{-3mm}
\paragraph{A nearly optimal separation.}
As mentioned previously, the problem $\dSSP$ provides an oracle separation between $d$- and $(2d+1)$-depth quantum circuits in the presence of polynomial-time classical computation. 
Next we further improve the separation to distinguish $d$- from $(d+3)$-depth hybrid quantum computation. 
In particular, we modify the problem to allow a $(d+3)$-depth algorithm to succeed with high probability, while at the same time, it remains hard for a $d$-depth prover to learn the hidden shift.

First we recall that a $(2d+1)$-depth quantum algorithm is needed because to simulate a query to $f$, the algorithm queries $f_0,f_1,\ldots,f_d$ followed by queries to $f_{d-1},\ldots,f_0$ to uncompute the intermediate values. 
To avoid the need of extra depth for uncomputation, our idea is to replace the standard access to $f_0,\ldots,f_{d}$ with ``in-place oracles.'' 
In this model, the algorithm is given access to $\ket{x} \xmapsto{f_i} \ket{f_i(x)}$, and thus the intermediate queries have been erased automatically. 
While in-place oracle access to an arbitrary is not a unitary in general, 
in our case, perhaps fortunately, the functions $f_0,\dots,f_d$ are either permutations or 2-to-1 functions. It is clear that the in-place oracle access for permutations is a unitary. 
Furthermore, we modify the last function $f_d$ such that $f_d$ is bijective, but a depth $(d+1)$-depth algorithm can simulate a query to the underlying Simon's function $f$ with constant probability.
We call the same problem with in-place oracle access the \emph{in-place} $d$-Shuffling Simons Problem (in-place $\dSSP$, see \dfn{ipdssp}).
Finally, we show that that the in-place $\dSSP$ cannot be solved by any hybrid quantum-classical computers with quantum circuit depth at most $d$.

\subsubsection{Single-prover protocol from LWE}

The second protocol relies on the assumption that the Learning-with-Errors (LWE) problem is hard for quantum computers (also called the QLWE assumption). 
In a breakthrough~\cite{Brakerski18}, Brakerski, Christiano, Mahadev, Vazirani and Vidick showed that the QLWE assumption implies the existence of a noisy trapdoor claw-free function (NTCF). %
A function $f$ is trapdoor claw-free if it is 2-to-1, and given a pair $(x,y)$ such that $f(x)=y$, it is computationally intractable to find the other preimage of $y$.
Furthermore, the function $f$ is also equipped with a strong property called the adaptive hardcore bit property.
In a nutshell, the property states that no quantum adversary given access to a description of $f$ can output $(y,x,e)$ such that $x$ is a preimage of $y$ and $e\cdot (x_0+x_1)=0$ with probability non-negligibly better than 1/2, where $x_0,x_1$ are the preimages of $y$.
In contrast, there exist quantum processes which allows an efficient quantum device to output either $(y,x)$ or $(y,e)$.

This observation leads to a proof-of-quantumness protocol: 
the verifier on receiving $y$ requests the prover to present a preimage $x$ or an equation $e$. An efficient quantum prover can succeeds with nearly perfect probability. For proving classical hardness, the idea is that one can rewind a classical prover which succeeds with probability non-negligibly more than 1/2 to extract both $x$ and $e$ with non-negligibly probability: For every prover $\A$, let the state before receiving the challenge be a random variable $\sigma_y$.
The adversary challenges $\A$ to use the same state $\sigma_y$ to output both a preimage and an equation. 
Any prover $\A$ that wins the test with non-negligible advantage would imply that the adversary breaks the property.

In subsequent works, Hirahara and Le Gall \cite{hirahara2021test} and Liu and Gheorghiu \cite{liu2021depth} showed that the same protocol only requires a quantum prover of constant depth.
The ideas behind these constructions basically follow from presenting NTCFs that can be evaluated with constant quantum depth.

A proof-of-quantumness protocol can be viewed as a protocol which separates a prover of non-zero quantum depth from one of zero quantum depth (i.e., a classical device).
It seems natural to rely on the same hardness assumption to separate a high-depth quantum device from a low-depth one with the following protocol:
\begin{enumerate}
    \item The verifier samples the functions $f_1,\ldots,f_d$ and sends these functions to the prover.
    
    \item The prover outputs $y_1,\ldots,y_d$.
    
    \item For $i=1\ldots d$, the verifier sequentially samples a random bit $c_i$ which indicates the request to send a preimage $x_i$ or a equation $e_i$ for $y_i$.
    The verifier rejects if in any of the rounds the prover fails.
\end{enumerate}
In this protocol, the prover must increase its quantum depth by 1 in each round of Step 3, since the operation the prover performs depend on the challenge bit $c_i$, which depends on the previous message from the prover.
It is straightforward to see a $(d_f+d)$-depth prover succeeds with nearly perfect probability, where $d_f$ is the depth required for the evaluation of $f$.
However, to show the hardness for any small-depth device, since the device is no longer purely classical, the same rewinding argument does not directly apply. %

We formalize an observation that a $(d-1)$-depth prover cannot stay coherent throughout the protocol, and has to ``reset'' (i.e., to destroy all its coherence and to continue with a purely classical state) in an intermediate round $j$. 
Thus from round $i=j\ldots d$, the prover begins with an intermediate classical state $\sigma$, and responds with its quantum power.
To break the adaptive hardcore bit property, the reduction simulates the protocol to compute the state $\sigma$, and rewinds on $\sigma$ to compute both a preimage and an equation for $f_d$.

\vspace{-3mm}
\subsection{Discussion and open problems}
We give protocols that allows a classical verifier to distinguish quantum machines with different circuit depths and polynomial-time classical computation. 
Our first two-prover protocol can achieve nearly optimal separation and information theoretic security by adding an additional untrusted quantum helper, and the verification can be made efficient if quantum-secure pseudorandom functions exist. 
The second protocol achieves a single-prover $\CVQD$ based on QLWE, and works for a slightly larger constant promise gap on the depth. %
We note that the two protocols we present in this paper are incomparable.
The first protocol makes no additional assumption to certify that the target machine runs in small quantum depth.
In contrast, the second requires no additional prover to achieve the same task based on a widely-held assumption. 

We include a few open questions.
First, the two-prover protocol has separation $d$ versus $d+3$ and the single-prover protocol has separation $d$ versus $d+d_f+1$. It is interesting to know if these separations can be improved.
Secondly, it would be interesting to know if we can directly instantiate $d$-SSP from (standard) computational assumptions. %
This is similar to the case for instantiating abelian hidden subgroup problems by factoring. 
If we can find such instantiation, then we can obtain a single-prover $\CVQD$ protocol that is different from the one in this work. 

In the single-prover $\CVQD$ protocol, we only know that an honest prover that implements ($d+d_f$)-QC schemes can succeed. It is open if a prover can convince the verifier by implementing a $(d+d_f)$-CQ scheme. %
Moreover, our single-prover verification scheme requires the use of randomized encoding of an NTCF family to achieve depth-efficient function evaluation, but since the number of qubits scales with the depth threshold, the honest prover would need a large space to succeed in the protocol.
Can we give a more space efficient protocol such that a demonstration of quantum depth can be implemented in a near-term quantum device?

Finally, the round complexity of our protocol scales with the depth.
In particular, to determine that a device has quantum depth no more than $d$, the round complexity is $O(d)$.
Can we give $\CVQD$ protocols for which the round complexity does not scale with the depth?

\vspace{-3mm}
\paragraph{Related work.} In an independent work~\cite{HLG22}, Atsuya Hasegawa and Fran{\c{c}}ois Le Gall defined the $d$-Bijective Shuffling Simon’s Problem that improves the quantum depth separation in~\cite{CCL19} to $d$ versus $d+1$ using the similar idea as in-place $d$-SSP (\dfn{ipdssp}). For in-place $d$-SSP, the gap is $d$ versus $d+1$ if we consider the same models as in Definition 3.8 and Definition 3.10 in~\cite{CCL19}. However, the models in~\cite{CCL19} count the depth of quantum queries to the oracle. In this work, we also count the two layers of Hadamard transforms at the beginning and the end of Simon's algorithm. This results in the gap $d$ versus $d+3$ in our first result in \thm{informal_1} (see \thm{cvqd_2p_opt} and \cor{cvqd_2p_final} for formal statements). 

\subsection{Organization}
The rest of the paper is organized  as follows. 
\sec{preliminaries} includes the required technical background knowledge for this paper, and our modifications of the previous protocols which will be useful for our contributions.
\sec{epr_protocol} defines a transformation from a quantum oracle separation to a two-prover protocol that preserves completeness and soundness. 
\sec{leash_protocol} presents a framework that transforms a quantum oracle separation to a two-prover protocol with a classical verifier. 
\sec{CVQD} shows a protocol for classical verification of quantum depth under the framework developed in \sec{epr_protocol} and \sec{leash_protocol}. 
Finally, in \sec{cqd-lwe}, we present a new single-prover protocol from QLWE.

\section*{Acknowledgement}

We thank Scott Aaronson, Kai-Min Chung, and anonymous reviewers for helpful suggestions on an earlier version of this paper.
We also thank Atsuya Hasegawa and Fran{\c{c}}ois Le Gall for sharing their results \cite{HLG22} with us.
SHH acknowledges the support from Simons Investigator in Computer Science award, award number 510817.

\section{Preliminaries}\label{sec:preliminaries}

For finite set $\X$, we denote $x\gets_R\X$ the process of sampling a random variable $x$ uniformly from $\X$.
For distribution $\D$ over a finite set $\X$, we denote $x\gets_R\D$ the process of sampling a random variable $x\in\X$ according to $\D$.
For a classical or quantum process $\A$, we denote $y\gets\A(x)$ to specify that $\A$ on input $x$ outputs $y$.
A function $f:\mathbb N\to\mathbb R$ is negligible, denoted $f(n)=\negl(n)$, if there exists an integer $n_0$ such that for $n\geq n_0$, $f(n)\leq n^{-c}$ for every constant $c$.
In other words, $f$ if negligible if $f(n)=n^{-\omega(1)}$.
We use the notation $1_P$ to denote 1 if $P$ is true and 0 if $P$ is false.

\subsection{Oracle separation for quantum depth}

We first introduce the two models for interleaving $d$-depth quantum circuits and classical polynomial-time computation.

\begin{definition}[$d$-CQ scheme~\cite{CCL19}]
\label{dfn:dcq}
Let $k = \poly(n)$. Let $\A^1_c,\dots,\A^k_c$ be a sequence of classical polynomial-time algorithms and $\A^1_q,\dots,\A^k_q$ be a sequence of $d$-depth quantum circuits. A $d$-CQ scheme can be represented as following: 
\begin{align*}
    \A^k_c\circ (\Pi_{0/1}\circ \A^k_q)\circ\cdots\circ \A^2_c\circ (\Pi_{0/1}\circ \A^2_q)\circ \A^1_c\circ (\Pi_{0/1}\circ \A^1_q), 
\end{align*}
where, $\Pi_{0/1}$ is a measurement on all qubits in the computational basis. 
\end{definition}

\begin{definition}[$d$-QC scheme~\cite{CCL19}]
\label{dfn:dqc}
Let $k = \poly(n)$. Let $\A_c^0,\A_c^1\dots,\A^d_c$ be a sequence of classical polynomial-time algorithms and $\A^1_q,\dots,\A^d_q$ be a sequence of $1$-depth quantum circuits. A $d$-CQ scheme can be represented as following: 
\begin{align*}
    \A^{d}_c\circ (\Pi_{0/1}\otimes I)\circ \A^d_q\circ\cdots\circ \A^2_c\circ (\Pi_{0/1}\otimes I)\circ \A^2_q\circ \A^1_c\circ (\Pi_{0/1}\otimes I)\circ \A^1_q\circ \A_c^0, 
\end{align*}
where, $\Pi_{0/1}\otimes I$ is a computational basis measurement that only operates on part of the qubits. The input of $\A_q^i$ includes the output quantum state of $\A_q^{i-1}$ for $i=2,\dots,d$ and the classical information from $\A^{j}_c$ and the measurement outcome of $\A^j_q$ for $j<i$.  The input of $\A_c^i$ includes the measurement outcome of $\A_q^{i}$ and other classical information from $\A^j_c$ and $\A^j_q$ for $j<i$ for all $i\in[d]$.
\end{definition}

\begin{remark}\label{rmk:gateset}
In this work, we generally choose the universal gateset to be all one- and two-qubit gates. In particular, the impossibility results in \thm{informal_1} showing all $d$-CQ and $d$-QC schemes fail the $\CVQD$ protocols hold for any universal gateset with bounded fan-in gates. 
\end{remark}

Roughly speaking, $d$-CQ schemes allow a classical algorithm to access a $d$-depth quantum circuit polynomially many times; however, all the qubits of the quantum circuit need to be measured after each access (and thus no quantum state can be passed to following $d$-depth quantum circuits). On the other hand, $d$-QC schemes let a quantum circuit to access classical algorithms after each depth and pass quantum states to the rest of the circuits for at most $d$ depths.

In~\cite{CCL19}, the class of problems that can be solved by $d$-CQ schemes is defined as $\BPP^{\BQNC_d}$, and the class of problems that can be solved by $d$-QC schemes is defined as $\BQNC_d^{\BPP}$. 

Chia, Chung and Lai~\cite{CCL19} presented an oracle problem that can separate schemes in~\dfn{dcq} and~~\dfn{dqc} with different quantum circuit depths. We briefly introduce the oracle separation in the following. 

\begin{definition}[$d$-shuffling {\cite[Definition~4.1]{CCL19}}]\label{dfn:d-shuffling} 
Let $f:\{0,1\}^n \rightarrow \{0,1\}^n$ be any function. A $d$-shuffling of $f$ is defined by $\F:=(f_0,\ldots,f_d)$, where $f_0,\ldots,f_{d-1}$ are random permutations over $\bit^{(d+2)n}$.
The last function $f_d$ is a fixed function satisfying the following properties: let $S_d:=\{f_{d-1}\circ\cdots\circ f_0(x'):x'\in\{0,1\}^n\}$.
\begin{itemize}
    \item For $x\in S_d$, let $f_{d-1}\circ\cdots\circ f_0(x')=x$, and choose the function $f_d:S_d\to[0,2^n-1]$ such that $f_d\circ f_{d-1}\circ\cdots\circ f_0(x')=f(x')$.
    \item For $x\notin S_d$, $f_d(x)=\perp$.
\end{itemize}
\end{definition}

Then, we recall the definition of Simon's function. 
\begin{definition}[Simon's function]
  For a finite set $S$ and $s\in\bit^n$ (also called the hidden shift), the Simon's function $f:\bit^n\to S$ satisfies that $f(x)=f(x')$ if and only if $x'=\{x,x\oplus s\}$.
\end{definition}

The Simon's problem is to compute the hidden shift $s$ given oracle access to a Simon's function $f$. The quantum algorithm for Simon's problem uses one quantum query to sample a random vector $y$ satisfying $y\cdot s=0$. Making $O(n)$ queries suffices to find a generating set of the subspace $H=\{y:y\cdot s=0\}$ with overwhelming probability, and thus the hidden shift is uniquely determined from the generators. It is worth noting that any classical algorithm that finds $s$ with high probability requires $\Omega(\sqrt{2^n})$ queries even if the Simon's function is given uniformly randomly. A random Simon's function is defined as a function drawn uniformly from the set of Simon's functions from $\bit^n$ to $S$ and we choose $S=\bit^n$.

We now define the $d$-Shuffling Simon's problem ($\dSSP$) that separates $\BPP^{\BQNC_{2d+3}} \cap \BQNC_{2d+3}^{\BPP}$ from $\BPP^{\BQNC_{d}} \cup \BQNC_{d}^{\BPP}$ relative to an oracle. 
\begin{problem}[$d$-shuffling Simon's problem ($\dSSP$) {\cite[Definition~4.9]{CCL19}}]\label{prob:dSSP}
Let $n\in\mathbb{N}$ and $f:\bit^n\to\bit^n$ be a random Simon's function.
Given oracle access to the $d$-shuffling $\F:=\{f_0,f_1,\dots,f_d\}$ of $f$, the problem is to find the hidden shift $s$ of $f$.
\end{problem}

Chia, Chung and Lai showed the following theorem \cite{CCL19}.

\begin{theorem}[\cite{CCL19}]
\label{thm:dssp_ccl19}
Let $d = \poly(n)$. The $d$-SSP problem can be solved by $(2d+3)$-CQ and $(2d+3)$-QC schemes with oracle access to the $d$-shuffling oracle of $f$. Furthermore, for any $d'$-CQ and $d'$-QC schemes $\A$ with with oracle access to the $d$-shuffling oracle of $f$ and $d'\leq d$, the probability that $\A$ solves the problem is negligible. 
\end{theorem}

\begin{remark}
\label{rmk:model}
In~\cite{CCL19}, it said that $d$-SSP can be solved by $(2d+1)$-CQ and -QC schemes because the models defined in Definition 3.8 and Definition 3.10 in~\cite{CCL19} mainly considered the depth for querying the oracle. Here, for our purpose, we count the two Hadamard transforms at the beginning and the end of Simon's algorithm, which gives additional two depths.
\end{remark}

This means that when there is a quantum algorithm of $(2d+3)$ quantum circuit depth (including access to the oracle) succeeding with probability at least 2/3 (in fact the success probability is $1-\negl(n)$). 
The second part of \thm{dssp_ccl19} shows that every quantum algorithm of quantum circuit depth at most $d$ outputs the hidden shift with negligible probability, even if it makes an arbitrary polynomial number of queries.

\subsection{Quantum-secure pseudorandom permutations}

For our task, we also want the oracle can be implemented efficiently. However, by a counting argument, a random permutation cannot be computed efficiently with overwhelming probability. We use pseudorandom permutations to address this issue. In a query model, we can replace a random permutation with a pseudorandom one without decreasing the performance of a query algorithm by non-negligible difference.
\begin{definition}[Quantum-secure pseudorandom permutations (qPRP) {\cite{Zha16}}]\label{dfn:qprp}
For security parameter $\lambda$ and a polynomial $m=m(\lambda)$, a pseudorandom permutation $P$ over $\bit^{m}$ is a keyed function $\K\times\bit^m\to\bit^m$ such that there exists a negligible function $\epsilon$ such that for every quantum adversary $\A$, it holds that
\begin{align}
    \Big|\Pr_{F\gets_R\P}[\A^{O_{F},O_{F^{-1}}}=1]-\Pr_{k\gets_R\K}[\A^{O_{P(k,\cdot)}, O_{P^{-1}(k,\cdot)}}=1]\Big|\leq \epsilon(\lambda),
\end{align}
where $\P$ is the set of permutations over $\bit^m$ and $O_Q:\ket{x,y}\mapsto\ket{x,y\oplus Q(x)}$ for permutation $Q:\bit^m\to\bit^m$ and $x,y\in\bit^m$.
\end{definition}

\subsection{The Broadbent protocol for verifying quantum computation}\label{sec:broadbent}

In this section, we briefly introduce the Broadbent protocol for verifying quantum computation \cite{Bro18}.
The protocol consists of two parties, the prover $P$ which is untrusted but can perform arbitrary quantum computation, and the verifier $V$ which is almost classical.
In particular, $V$ can perform measurements in certain bases.
The prover and the verifier interact, and at the end of the protocol, the verifier outputs a bit which is either ``accept'' or ``reject.''
The protocol can be used to verify a complete language in $\BQP$ (more precisely, $\mathsf{PromiseBQP}$):
\begin{itemize}
    \item Completeness: if the computation $U$ satisfies $\|\Pi_0U\ket{0}\|^2\geq 2/3$, then there exists a quantum prover which makes $\verifier$ accept with probability at least $c$.
    
    \item Soundness: if the computation $U$ satisfies $\|\Pi_0 U\ket{0}\|^2\leq 1/3$, then for every prover, the verifier accepts with probability no more than $s$.
\end{itemize}
Here the projector $\Pi_0=\proj{0}\otimes\Id$ refers to the event that measuring the first qubit of the state $U\ket{0}$ in the standard basis yields an outcome $0$.
The parameters $c,s$ are called the completeness and soundness respectively.

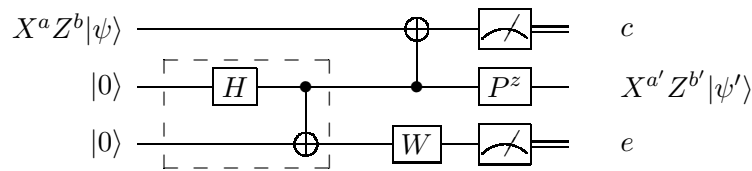
\begin{figure}[htbp]
    \centering
    \begin{minipage}{.8\linewidth}
    \Qcircuit @C=1.3em @R=0.6em @!R {
    \lstick{X^aZ^b\ket{\psi}} & \qw & \qw & \qw & \qw & \targ & \meter & \cw & \rstick{c} \\
    \lstick{\ket{0}} & \push{\rule{0em}{1.1em}}\qw & \gate{H} & \ctrl{1} & \qw & \ctrl{-1} & \gate{P^z}  & \qw & \rstick{X^{a'}Z^{b'} \ket{\psi'}} \\
    \lstick{\ket{0}} & \qw & \qw & \targ & \qw & \gate{W} & \meter & \cw & \rstick{e}
    \gategroup{2}{2}{3}{4}{.8em}{--}
    }
    \end{minipage}
    \caption{The quantum circuit of the $T$ gadget. The dashed box prepares an EPR pair. The first two qubits are held by the prover and the third qubit is held by the verifier. The bit $z\gets_R\bit$ is sampled independently for each $T$ gadget by the verifier. For the choice of $W$, see \tabl{W}.}
    \label{fig:t-gadget}
\end{figure}

    \begin{table}[htbp]
        \centering
        \begin{tabular}{|c|c|}
        \hline
            Round type & Unitary $W$ \\
            \hline\hline
             Computation & $HP^{a'+c+z}T$ \\
             Even parity & $\Id$ \\
             Odd parity & $HP^z$\\
             \hline
        \end{tabular}
        \caption{The choice of unitary $W$ in a $T$ gadget. The table is modified from \cite[Table~3]{CGJV19}.
        }
        \label{tabl:W}
    \end{table}

In the Broadbent protocol, the prover and the verifier share $(n+t)$ EPR pairs, where $n$ is the number of qubits the computation $U$ acts on, and $t$ is the number of $T$ gadgets in $U$.
The quantum circuit for the $T$ gadget is shown in \fig{t-gadget}.
The computation $U$ is a quantum circuit consisting of gates in $\{CNOT,H,T\}$ (which is a universal gate set), performed on a quantum state encrypted by quantum one-time pad.
The prover is designated to perform the following operations for each gate.
\begin{itemize}
\item For each $CNOT$ gate, the prover is designated to perform the gate on the associated qubits.

\item For each $T$ gate, the prover performs a $T$ gadget using a half of an EPR pair as an ancilla qubit.

\item Each Hadamard gate is compiled into the sequence $HTTHTTHTTH$ of single-qubit gates.
For each $T$ gate in the sequence, they run a $T$ gadget.
For each Hadamard gate in the sequence, the prover applies a Hadamard gate.
We will also refer to the implementation as a Hadamard gadget (or an $H$ gadget for short).
\end{itemize}
The verifier $V$ chooses to run one of the following rounds, and performs the operations on the other halves of the $(n+t)$ EPR pairs depending on the choice of rounds.
More specifically, $\verifier$ performs measurements on the EPR pairs to create the initial state and the states required for each $T$ gadget, and the actions for each round type is specified as follows.

    \begin{table}[]
        \centering
        \begin{tabular}{|c|c|}
            \hline
            Gate/Gadget & Key update rule \\
            \hline\hline
             $T$ gadget, computation & $(a+c,b+e+a+c+(a+c)z)$ \\
             $T$ gadget, even parity & $(e,0)$ \\
             $T$ gadget, odd parity & $(0,b+e+z)$ \\ 
             $H$ gate & $(b,a)$\\
             $CNOT$ gate & $(a,b+b',a+a',b')$ \\
             \hline
        \end{tabular}
        \caption{Rules for updating the one-time-pad keys after applying each gate in the protocol. The bit $z$ is chosen uniformly by the verifier, and the measurement outcomes $e,c$ are obtained by the prover (see \fig{t-gadget}). For the $T$ gadget and the Hadamard gate, before updating the key is $(a,b)$. For a CNOT gate, the keys are $(a,b)$ on the control qubit and $(a',b')$ on the target.
      }
        \label{tabl:key-update}
    \end{table}
    
\begin{itemize}
    \item Computation round (with constant probability $p$): the computation is delegated to the prover $P$. If $P$ behaves honestly, the result can be recovered by $V$.
    In the beginning of the protocol, $V$ measures the first $n$ EPR pairs in the standard basis.
    The half held by the prover collapses to a quantum one-time pad of the zero state.
    Then the (honest) prover performs the quantum computation on this state. 
    For each $T$-gadget (used to apply a $T$ gate or a Hadamard gate), the prover uses the a half of the EPR pair associated with the $T$-gadget as an ancilla qubit, and the verifier performs the associated measurement so the $T$-gadget implements a $T$ gate, up to a key update (\tabl{key-update}).

    \item Test round (with probability $1-p$): $V$ tests if $P$ behaves honestly, and rejects if any error is detected.
    A test round has two types, each of which is executed with probability 1/2, outlined as follows.
    \begin{itemize}
        \item An $X$-test round is used to detect if there is a bit flip error.
        In an $X$-test round, the verifier measures the first $n$ EPR half in the standard basis to create a quantum one-time pad of the zero state $\ket{0}^{\otimes n}$.
        For each $T$-gadget, the verifier performs the measurement in the basis such that the $T$-gadget acts trivially (as the identity operation) up to a key update (\tabl{key-update}).
        In the end, the verifier applies the key update rule to compute the key, and decrypts the first qubit.
        The verifier accepts if the bit is $0$ and rejects otherwise.
        
        \item A $Z$-test round is used to detect if there is a phase flip error.
        The operations are the same as the $X$-test round except that they are performed in the Hadamard basis.
        The verifier measures the first $n$ EPR pairs in the Hadamard basis to create a quantum one-time pad of the plus state $\ket{+}^{\otimes n}$.
        For each $T$ gadget, the verifier performs the measurement in some basis such that the $T$ gadget acts trivially up to a key update (\tabl{key-update}).
        In the end, the verifier applies the key update rule to compute the key, and decrypts the first qubit.
        The verifier disregards the result and rejects only if any error was detected throughout the computation.
        
    \end{itemize}
\end{itemize}
The protocol performs quantum computation on encrypted data by quantum one-time pad.
Depending on the round type, different key update rules are adopted.
For a Hadamard gadget, six $T$-gadgets are performed. 
Recall that in an $X$-test round, a $T$-gadget acts trivially on the zero state. Passing an odd number of Hadamard gates yields an encrypted plus state, and thus the $T$-gadget used in this case will be the same as a $T$ gate application in a $Z$-test round.
Thus for convenience, we may define the parity of a $T$-gadget as follows. %
A $T$-gadget is of even parity if it is not part of an Hadamard gadget, or an even (resp. odd) number of Hadamard gates has been applied before in an Hadamard gadget in an $X$-test (resp. a $Z$-test) round; otherwise it is of odd parity.
The key update rules are formally defined in \tabl{key-update}.

The original Broadbent protocol is used to verify $\BQP$ languages (more precisely $\mathsf{PromiseBQP}$).
Now we consider the following modification to show a rigidity result: 
the last message from the prover to the verifier is an $n$-qubit quantum state $\rho$.
The verifier computes the key according to the key update rule (\tabl{key-update}), and applies one-time pad decryption on $\rho$, with the following modification after the last message from the prover is sent.
\begin{itemize}
    
    \item In an $X$-test round, the verifier accepts if measuring the state in the standard basis yields $0^n$.
    
    \item In a $Z$-test round, the verifier accepts if measuring the state in the Hadamard basis yields $0^n$.
\end{itemize}
The differences are instead of checking only the first qubit, the verifier determines, in a test round, if an attack has been applied on any of the qubits.
The new verification procedure requires a larger quantum channel for the prover to send the entire state to the verifier.
As we will see, the extra cost is not necessary when the protocol is made purely classical with two provers.

Our rigidity statement is formally stated as follows:
if the prover is accepted with probability $1-\epsilon$ in test rounds, 
the prover in computation run implements a quantum channel $\E_C$ that is $O(\epsilon)$-close to the honest computation.
First recall the following fact about Pauli twirls.
\begin{lemma}[Pauli twirls]\label{lem:pauli-twirl}
  Let $P_a:=X^{a_1} Z^{a_2}$ denote a Pauli operator for $a=(a_1,a_2)$ where $a_1,a_2\in\bit^n$.
  For any quantum state $\rho$ and quantum channel $\Phi$, it holds that 
  \begin{align}
    \sum_{a\in\bit^{2n}} P_a^\dag \Phi(P_a \rho P_a^\dag) P_a = \sum_a r_a P_a\rho P_a^\dag,
  \end{align}
  for some distribution $r$ over Pauli matrices.
\end{lemma}

Broadbent shows that without loss of generality, any attack performed by the prover can be written as a honest execution $C$ followed by a quantum channel $\Phi$. %
The prover first performs $C$ and yields a quantum state $\rho=C\proj{\psi}C^\dag$, one-time padded with key $Q$. The prover then applies any quantum channel $\Phi$, followed by decryption performed by the verifier.
By \lem{pauli-twirl}, the prover's attack can be written as a probabilistic mixture of Pauli unitaries.
We prove the following theorem.

\begin{theorem}\label{thm:broadbent}
  For $\epsilon\in[0,1/2]$, any prover who succeeds with probability $1-\epsilon$ in the test rounds if and only if it implements a quantum channel $\E_C$ satisfying $\|\E_C-\C\|_\diamond\leq 4\epsilon$ where $\C(\rho):=C\rho C^\dag$ in the computational round.
\end{theorem}
\begin{proof}
  Without loss of generality, we may express the action of any prover given access to circuit description $C$ and the verifier's decryption as a quantum channel
\begin{align}
  \E_{C,i} = \Exp_k (\O_k \circ \Phi \circ \O_k) \circ \C_i = \P\circ \C_i,
\end{align}
for some Pauli channel $\P$ with Kraus form $\{r_a^{1/2}P_a : a\in\bit^{2n}\}$ by \lem{pauli-twirl} and round type $i$.
Also recall that here we denote $P_a=X^{a_1}Z^{a_2}$ for $a=(a_1,a_2)$ and $a_1,a_2\in\bit^n$.
Since these rounds look completely identical to the prover, the quantum channel $\Phi$ must be identical among the choices of round type. 

Let the success probability of the prover in an $X$- and a $Z$-test round be $1-\epsilon_X$ and $1-\epsilon_Z$ respectively.
If the prover succeeds with probability $1-\epsilon$ conditioned on the event that a test round is chosen, 
it must hold that $\epsilon_X,\epsilon_Z\leq 2\epsilon$ since otherwise $\frac{1}{2}(1-\epsilon_X)+\frac{1}{2}(1-\epsilon_Z)=1-\frac{\epsilon_X+\epsilon_Z}{2}< 1-\epsilon$.

In an $X$-test round, $\C_X=\I$, the identity channel, and 
\begin{align}\nonumber
    1-\epsilon_X 
  &= \bra{0^n}\P(\proj{0^n})\ket{0^n} \\\nonumber
  &= \sum_a r_a |\bra{0^n} P_a \ket{0^n}|^2 \\
  &= \sum_{a:a_1=0} r_a \geq 1-2\epsilon.
\end{align}
Similarly, in a $Z$-test round, $\sum_{a:a_2=0}r_a\geq 1-2\epsilon$. 
Combining the inequalities, we conclude that $r_0\geq 1-4\epsilon$.
Therefore, for any prover who succeeds with probability $1-\epsilon$, it holds that in the computation round 
\begin{align}\nonumber
  \|\P\circ\C-\C\|_\diamond 
  &\leq \|\P-\I\|_\diamond \\\nonumber
  &= \max_{\rho:\tr\rho=1}\tr\left|\sum_a r_a (P_a\rho P_a^\dag-\rho)\right| \\\nonumber
  &\leq \sum_{a\neq 0}r_a\|\P_a-\I\|_\diamond \\
  &\leq 4\epsilon,
\end{align}
where $\P_a(\rho)=P_a\rho P_a^\dag$.
The above reasoning shows the ``only if'' direction.

If the prover implements $\E_C=\P\circ\C$ which is $\delta$-close to $\C$ in the computation round.
Then in the $X$-test round, we have 
$\|\P(\proj{0^n})-\proj{0^n}\|_\tr\leq \|\E_C-\C\|_\diamond\leq\delta$.
Similarly, in the $Z$-test round, we have 
$\|\P(\proj{+^n})-\proj{+^n}\|_\tr\leq \|\E_C-\C\|_\diamond\leq\delta$.
Then we conclude that the success probability is at $1-\delta$ when a $X$-test or a $Z$-test round is chosen.
This implies that the success probability in a test round is at least $1-\delta$.
\end{proof}

\subsection{The Verifier-on-a-Leash protocol}\label{sec:leash}

The Broadbent protocol can be used to verify arbitrary quantum computation, but it requires the verifier to have the capability to perform measurements in the bases listed in \tabl{W}.
To achieve purely classical verification, Coladangelo, Grilo, Jeffery and Vidick presented a new protocol called the Verifier-on-a-Leash protocol (or the Leash protocol).
In the Leash protocol, two provers called PP and PV share entanglement, but are not allowed to communicate with each other after the protocol starts. 
PP plays the role of the prover in the Broadbent protocol, and PV performs the measurements by the verifier in the Broadbent protocol.
To certify PV plays honestly, a new rigidity test is introduced to verify Clifford measurements, in particular, the observables in $\Sigma=\{X,Y,Z,F,G\}$ where $X,Y,Z$ are Pauli matrices, $F=\frac{1}{\sqrt 2}(-X+Y)$ and $G=\frac{1}{\sqrt 2}(X+Y)$.
The idea is to use a non-local game for certifying the standard basis measurement (i.e., the observable $X$) and the Hadamard basis measurement (the observable $Z$).
We will modify a test called $\rigid'(\Sigma,m)$ which certifies Clifford measurements by Coladangelo, Grilo, Jeffery and Vidick \cite{CGJV19} to only act on a subset of qubits. %
In particular, the following theorem holds with the test.
\begin{theorem}[{\cite[Theorem~4]{CGJV19}}]\label{thm:cliff}
  There exists a test $\rigid'(\Sigma,m)$ such that the following holds:
  suppose a strategy for the players succeeds in test $\rigid'(\Sigma,m)$ with probability at least $1-\epsilon$.
  Then for $D\in\{A,B\}$, there exists an isometry $V_D$ such that 
  \begin{align}
      \|(V_A\otimes V_B)\ket{\psi}_{AB} - \ket{EPR}^{\otimes m}_{A'B'}\ket{AUX}_{\hat A\hat B}\|^2\leq O(\sqrt\epsilon),
  \end{align}
  and
  \begin{align}
      \Exp_{W\in\Sigma^m}
      \sum_{u\in\{\pm\}^m}
      \left\|
      V_A \tr_B((\Id_A\otimes W_B^u)\proj{\psi}_{AB} (\Id_A\otimes W_B^u)) V_A^\dag
      -
      \sum_{\lambda\in\{\pm\}}
      \left( \bigotimes_{i=1}^m \frac{\sigma_{W_i,\lambda}^{u_i}}{2} \otimes \tau_\lambda \right)
      \right\|_1
      = O(\poly(\epsilon)).
  \end{align}
  Moreover, players employing the honest strategy succeed with probability $1-e^{-\Omega(m)}$ in the test.
\end{theorem}

\thm{cliff} shows that a strategy that succeeds in $\rigid'(\Sigma,m)$ with probability at least $1-\epsilon$ must satisfy the following conditions:
The players' joint state is $O(\sqrt\epsilon)$-close to a tensor product of $m$ EPR pairs together with an arbitrary ancilla register.
Moreover, on average over uniformly chosen basis $W\gets_R\Sigma^m$, the provers' measurement is $\poly(\epsilon)$-close to a probabilistic mixture of ideal measurements and their conjugates. 
More specifically, the probabilistic mixture can be realized as follows: performing a measurement on the ancilla register yields a post-measurement state $\tau_\lambda$ and a single bit outcome $\lambda\in\{\pm\}$ that specifies whether the ideal measurement associated with the observable $\sigma_{W_i,+}$ or that with its conjugate $\sigma_{W_i,-}$ is performed.\footnote{The conjugate of a matrix $A$ is obtained by taking the complex conjugate of each element of $A$.}
  Then applying $\sigma_{W_i,\lambda}$ yields a post-measurement state $\sigma_{W_i,\lambda}^{u_i}$ and an outcome $u_i\in\{\pm\}$ for each index $i\in[m]$.
We can only hope to certify that the strategy is close a probabilisitc mixture of the ideal strategy and its conjugate because a protocol with classical communication does not distinguish a strategy from its conjugate (more concretely, replacing $i=\sqrt{-1}$ with $-i$ for any strategy by the prover does not change the score). 

In \sec{leash_protocol}, we will modify $\rigid'(\Sigma,m)$ to only apply on a random subset of indices in $[m]$, chosen by the verifier.
It is clear that for the subset, \thm{cliff} holds.

\subsection{Proof of quantumness}

Our protocol in the plain model will be based on a construction of noisy trapdoor claw-free function (NTCF) function family \cite{Brakerski18}, based on QLWE.
\begin{definition}[NTCF family {\cite[Definition~3.1]{Brakerski18}}]
Let $\lambda$ be a security parameter and $\X,\Y$ be finite sets.
Let $\K_\F$ be a finite set of keys.
A family of functions 
\begin{align}
    \F=\{f_{k,b}:\X\to\D_\Y\}_{k\in\K_\F,b\in\bit}
\end{align}
is called a noisy trapdoor claw free (NTCF) family if the following conditions hold:
\begin{enumerate}
    \item Efficient function generation.
    There exists an efficient probabilistic algorithm $\Gen_\F$ which generates a key $k\in\K_\F$ together with a trapdoor $t$: $(k,t)\gets\Gen_\F(1^\lambda)$.
    
    \item Trapdoor injective pair.
    For all keys $k\in\K_\F$, the following conditions hold.
    \begin{enumerate}
        \item Trapdoor:
        There exists an efficient deterministic algorithm $\Inv_\F$ such that for all $b\in\bit$, $x\in\X$ and $y\in\Supp(f_{k,b}(x))$, $\Inv_\F(t,b,y)=x$.
        Note that this implies that for all $b\in\bit$ and $x\neq x'$,
        $\Supp(f_{k,b}(x))\cap\Supp(f_{k,b}(x'))=\emptyset$.
        
        \item Injective pair: there exists a perfect matching $\R_k\subseteq\X\times\X$ such that $f_{k,0}(x_0)=f_{k,1}(x_1)$ if and only if $(x_0,x_1)\in\R_k$.
    \end{enumerate}
    
    \item Efficient range superposition. For all key $k\in\K_\F$ and $b\in\bit$, there exists a function $f_{k,b}':\X\to\D_\Y$ such that the following hold.
    \begin{enumerate}
        \item For all $(x_0,x_1)\in\R_k$ and $y\in\Supp(f_{k,b}'(x_b))$, $\Inv_\F(t,0,y)=x_0$ and $\Inv_\F(t,1,y)=x_1$.
        
        \item There exists an efficient deterministic procedure $\Chk_\F$ that, on input $k,b\in\bit$, $x\in\X$ and $y\in\Y$, returns 1 if $y\in\Supp(f_{k,b}'(x))$ and 0 otherwise.
        Note that $\Chk_\F$ is not provided the trapdoor $t$.
        
        \item For every $k$ and $b\in\bit$, 
        \begin{align}
            \Exp_{x\gets_R\X}[H^2(f_{k,b}(x),f_{k,b}'(x))] \leq \mu(\lambda),
        \end{align}
        for some negligible function $\mu$.
        Here $H^2$ is the Hellinger distance.
        Moreover, there exists an efficient procedure $\Samp_\F$ that on input $k$ and $b\in\bit$, prepares the state
        \begin{align}
            |\X|^{-1/2} \sum_{x\in\X,y\in\Y} 
            \sqrt{f_{k,b}'(x)(y)}\ket{x}\ket{y}.
        \end{align}
    \end{enumerate}
    
    \item Adaptive hardcore bit.
    For all keys $k\in\K_\F$, the following conditions hold for some integer $w$ that is a polynomially bounded function of $\lambda$.
    \begin{enumerate}
        \item For all $b\in\bit$ and $x\in\X$, there exists a set $G_{k,b,x}\subseteq\bit^w$ such that $\Pr_{d\gets_R\bit^w}[d\notin G_{k,b,x}]$ is negligible, and moreover there exists an efficient algorithm that checks for membership in $G_{k,b,x}$ given $k,b,x$ and the trapdoor $t$.
        
        \item There is an efficiently computable injection $J:\X\to\bit^w$ such that $J$ can be inverted efficiently on its range, and such that the following holds.
        If 
        \begin{align*}
            H_k &= \{(b,x_b,d,d\cdot(J(x_0)\oplus J(x_1))) | b\in\bit, (x_0,x_1)\in\R_k, d\in G_{k,0,x_0}\cap G_{k,1,x_1}\}, \\
            \bar H_k &= \{ (b,x_b,d,c) | (d,x,d,c\oplus 1)\in H_k \},
        \end{align*}
        then for any quantum polynomial-time procedure $\A$, there exists a negligible function $\mu$ such that 
        \begin{align}\label{eq:ahb}
            \left| \Pr_{(k,t)\gets\Gen_\F(1^\lambda)}[\A(k)\in H_k] - \Pr_{(k,t)\gets\Gen_\F(1^\lambda)}[\A(k)\in \bar H_k]\right| \leq \mu(\lambda).
        \end{align}
    \end{enumerate}
\end{enumerate}
\end{definition}
In a breakthrough, Brakerski, Christiano, Mahadev, Vazirani and Vidick give a proof-of-quantumness protocol (the BCMVV protocol) based on NTCFs \cite{Brakerski18}.
The protocol proceeds in the following steps.
\begin{enumerate}
    \item The verifier samples $(k,t)\gets\Gen(1^\lambda)$ and sends $k$ to the prover.
    
    \item The prover performs $\Samp_\F$ on the input state $\ket{+}$ and measures the image register to yield an outcome $y$.
    
    \item The verifier samples a random coin $c\gets_R\bit$ and sends $c$ to the prover.
    
    \item If $c=0$, the prover performs a standard basis measurement; otherwise the prover performs a Hadamard basis measurement.
    The outcome $w$ is then sent to the verifier.
    
    \item The verifier outputs $V(t,c,w)$, which is defined as
    \begin{align}\label{eq:bcmvv-v}
        V(t,y,c,w) :=
        \left\{
        \begin{array}{ll}
             1 & \text{if $c=0$ and $\Chk_\F(k,y,b,x)=1$, $w=(b,x)$} \\
             1 & \text{if $c=1$, $d\in G_{k,y}$ and $d\cdot(J(x_0)\oplus J(x_1))=u$, $w=(u,d)$} \\
             0 & \text{otherwise,}
        \end{array}
        \right.
    \end{align}
    where $G_{k,y}:=G_{k,0,x_0}\cap G_{k,1,x_1}$.
    Here we note that for $c=0$, $\Chk_\F$ does not require the trapdoor $t$, but the bit can be determined using the trapdoor.
\end{enumerate}
The adaptive hardcore bit property (see \eq{ahb}) implies that every adversary $\A$ given access to $k$, outputs $y,w_0,w_1$ such that $V(y,0,w_0)=V(y,1,w_1)=1$ with probability at most $1/2+\negl(\lambda)$. 
On the other hand, there exist efficient quantum processes to output valid $(y,w_0)$ or $(y,w_1)$ with probability $1-\negl(\lambda)$.

Hirahara and Le Gall \cite{hirahara2021test} and Liu and Gheorghiu \cite{liu2021depth} independently proposed two different methods of transforming the BCMVV protocol \cite{Brakerski18} into one that can be computed using only constant quantum depth (interleaving with classical computation).
In our second protocol for certifying quantum depth based on LWE, we will be using the theorem which states that there exists a construction of NTCF family for which the function evaluation takes constant depth, based on randomized encoding, defined as follows.

\begin{definition}[Randomized encoding \cite{AIK06,liu2021depth}]\label{dfn:randomized-encoding}
  Let $f:\bit^n\to\bit^\ell$ be a function and $r\gets_R\bit^m$.
  A function $\hat f:\bit^n\times\bit^m\to\bit^s$ is a $\delta$-correct, $\epsilon$-private randomized encoding of $f$ if it satisfies the following properties:
  \begin{itemize}
  \item Efficient generation: there exists a deterministic polynomial-time algorithm that, given a description of the circuit implementing $f$, outputs a description of a circuit implementing $\hat f$.
  \item $\delta$-correctness: there exists a deterministic polynomial-time algorithm $\textsc{Dec}$, called the decoder, such that for every input $x\in\bit^n$, $\Pr_r[\textsc{Dec}(\hat f(x,r))\neq f(x)]\leq\delta$. 
  \item $\epsilon$-privacy: there exists a PPT algorithm $S$, called the simulator, such that for every $x\in\bit^n$, the total variation distance between $S(f(x))$ and $\hat f(x,r)$ is at most $\epsilon$.
  \end{itemize}
Furthermore, a perfect randomized encoding is one for which $\epsilon=\delta=0$.
\end{definition}

\begin{theorem}[{\cite[Section~3.1]{liu2021depth}}]\label{thm:lg21-completeness}
  There exists an efficient quantum process which uses gates of bounded fan-out in constant quantum depth and prepares the state
  \begin{align}
      \sum_{b,x} \ket{b}\ket{x}\ket{\hat f_k(b,x)},
  \end{align}
  where $\hat f$ is a perfect randomized encoding of an NTCF $f$.
\end{theorem}
In particular, as shown in the same paper \cite[Section~3.3]{liu2021depth}, the proof-of-quantumness protocol takes total quantum depth 14 and three quantum-classical interleavings.
\begin{theorem}[{\cite[Theorem~3.1]{liu2021depth}}]\label{thm:lg21-soundness}
  There exists a perfect randomized encoding of an NTCF, which satisfies the randomness reconstruction property and is an NTCF.
\end{theorem}
In particular, \thm{lg21-soundness} implies that with the new construction, the adaptive hardcore bit property holds.

\section{Certifying a query algorithm}
\label{sec:epr_protocol}

A $q$-query algorithm $\A$ given access to $O$, denoted $\A^O$, without loss of generality, consists of a sequence of unitary maps
$U_q O \ldots U_1 O U_0$
acting on the zero state, followed by a standard basis measurement to extract the information about the oracle.
In the real world, information is transmitted in a non-black-box way, and therefore one must instantiate the oracle with an efficient quantum process.
However, any instantiation in the plain model would in general leak the information about the oracle, even based on computational assumptions \cite{BGIRSVY01}.
Thus, the analysis or the conclusions obtained in a query model does not apply when 
the oracle is replaced with an instantiation.

On the other hand, the impossibility result does not apply when the oracle is implemented by an external device to which the algorithm has limited access.
In this setting, another device, denoted $\prover_O$, is added to implement the oracle $O$, and the secret encoded in the oracle may be transparent to $\prover_O$.
An equivalent two-player protocol that mimics the query computation can be described as the following protocol:
\begin{enumerate}
\item A player $\prover_A$ performs $U_0$ on the zero state, and teleports the state to the other player $\prover_O$.
\item $\prover_O$ performs $O$ on the quantum state and teleport the output state back.
\item $\A$ performs $U_1$ on the input state and teleports the state.
\item Repeat the last two steps $(q-1)$ times with unitaries $U_2,\ldots,U_q$ and outputs $w$.
\item The verifier accepts if the oracle $O$ and $w$ satisfies some relation $R$ and rejects otherwise.
\end{enumerate}

When the players behave honestly, i.e., $\prover_O$ performs $O$ in each iteration, and $\prover_A$ performs $U_0,\ldots,U_q$ in order, the verifier accepts with probability exactly the same as the query algorithm $\A$ solving $R$.
When neither of the players is trusted, we ask the question: given a query algorithm $\A$, can we construct an equivalent protocol against untrusted players?
To answer the question, we specify the classes of problems which we consider in this work, and the equivalence between a query algorithm and a two-player protocol.

\paragraph{Oracle separation problems.}
For quantum complexity classes $\C_{yes},\C_{no}$, an oracle separation problem $R$ is a relation between the oracle $O$ (drawn from a public distribution) and a binary string $w$ such that %
\begin{itemize}
\item there is a quantum query algorithm $\A$ that is given access to $O$ and runs in $\C_{yes}$, outputs $w$ such that $R(O,w)=1$ with probability at least $p$, and
\item no quantum algorithm that is given access to $O$ and runs in $\C_{no}$, outputs $w$ such that $R(O,w)=1$ with probability at most $p'$.
\end{itemize}
Here, the probabilities $p$ and $p'$ are defined by $\C_{yes}$ and $\C_{no}$, e.g., if $\C_{yes} = \class{PP}$ and $\C_{no} = \class{BPP}$, $p=1/2$ and $p'=1/3$.

\paragraph{Two-player protocol for separation.}
A two-player protocol distinguishes $\C_{yes}$\footnote{In this work, we only focus on the machines that can at least teleport quantum states.} from $\C_{no}$ if there exists a classical verifier $\verifier$ interacting with two provers $\prover_O$ and $\prover_A$ such that the following conditions hold.
\begin{itemize}
\item \textbf{Non-locality.} $P_O$ and $P_A$ share an arbitrary quantum state and there is no quantum and classical channel between once the protocol starts. 

\item \textbf{Completeness.} There exists $\prover_O$ and $\prover_A$ that runs in $\C_{yes}$ such that 
$\Pr[\langle V,P_O,P_A\rangle=accept]\geq c$.

\item \textbf{Soundness.} For every $\prover_A$ that runs in $\C_{no}$ and any $\prover_O$, 
$\Pr[\langle V,P_O,P_A\rangle=accept]\leq s$.
Note that the protocol is sound for unbounded $P_O$.
\end{itemize}

\paragraph{Transforming oracle separation into two-player protocols.}
Our goal is to transform an oracle separation problem $R$ into a two-player protocol that distinguishes $\C_{yes}$ from $\C_{no}$.
In this work, we focus on the cases where the oracle $O$ is quantum accessible,  %
and the predicate $R$ can be computed deterministically. %
For real numbers $s\geq 0$ and $c\geq s+n^{-O(1)}$, we say a two-player protocol is $(c,s)$-equivalent to a problem $R$ if there exists a classical verifier $\verifier$ such that the following conditions hold.
\begin{itemize}
\item \textbf{Completeness.} For every $\A^O$ that runs in $\C_{yes}^O$ and solves $R$ with probability $2/3$, there exist 
$\prover_A$ that runs in $\C_{yes}$ and a quantum prover $P_O$, such that $\verifier$ accepts with probability at least $c$.
\item \textbf{Soundness.} For every $\prover_A$ that runs in $\C_{no}$, if $\verifier$ accepts with probability more than $s$, there exists an oracle algorithm $\A^O$ that runs in $\C_{no}^O$ and solves $R$ with probability more than $1/3$.

\end{itemize}
The completness statemenet together with the contrapositive of the soundness statement implies that if there is an oracle separation problem between $\C_{yes}$ and $\C_{no}$, then there is a two-player protocol that separates $\C_{yes}$ and $\C_{no}$.
We note the thresholds 2/3 and 1/3 are unimportant for our analysis; they can be replaced for giving a contradiction.

\section{A two-player non-local game for oracle separation}
\label{sec:leash_protocol}

In this section, we show how to classically verify a query algorithm.
More specifically, we show how to turn a query algorithm into a two-player protocol that preserves the completeness and soundness, up to an inverse polynomial multiplicative factor.
Our protocol consists of two quantum provers $\prover_A$ and $\prover_O$ sharing entanglements and interacting with a purely classical verifier $\verifier$.
The prover $\prover_O$ is designated to perform quantum computation $O$, and the prover $\prover_A$ performs arbitrary quantum computation to learn the information about $O$.

The protocol consists of two phases: in a query phase, $\prover_O$ receives a teleported quantum state from $\prover_A$.
An honest $\prover_O$ performs the unitary map of the oracle $O$ and teleports the resulting state to $\prover_A$.
Next, in the computation phase, $\prover_A$ may perform any quantum computation, subject to its resource constraints.

To classically verify whether $\prover_A,\prover_O$ instantiates a query algorithm in the two-player model, 
it is crucial that the verifier $\verifier$ can classically verify the $\prover_O$ behaves as intended.
To this end, we rely on the Leash protocol by Coladangelo, Grilo, Jeffery and Vidick \cite{CGJV19}, as introduced in \sec{leash}.
In this protocol, the classical verifier interacts with two entangled provers PP and PV.
PP is designated to perform a quantum computation $U$ on encrypted classical input $\ket{x}$ (using quantum one-time pad).
The Leash protocol relies on the Broadbent protocol to verify that PP performs $U$, conditioned on PV performing the correct single-qubit measurement for each $T$-gadget.
To classically verify PV's measurements, the verifier chooses either to run either a rigidity test to verify PV, or to execute the Broadbent protocol.
It is crucial that a rigidity test looks indistinguishable to PV from the real execution of the Broadbent protocol,
while at the same time PP needs to play differently for these two tests.
On the other hand, when running the Broadbent protocol, PP and PV are required to perform one of three indistinguishable tests introduced in the previous sections.
This allows the verifier to make sure PP behaves as intended.

To verify a query algorithm, a naive approach is to run the Leash protocol to check in every query, $\prover_O$ honestly run the unitary $O$, and $\prover_A$ performs the correct measurements (for the $T$-gadgets).
However, note that in the Leash protocol, the tests hide the round type from the provers, and only one ``query'' to PP is made.
Between the queries, $\prover_A$ performs any computation on the plaintext (which can be done by applying the Pauli correction on the teleported state from $\prover_O$).
When integrating the Leash protocol for multiple queries, the round type could possibly be leaked, at least to $\prover_O$: 
imagine that $\prover_A$ teleports its quantum state $\ket{\psi}$ to $\prover_O$.
When the verifier chooses to run a rigidity test, $\prover_O$ actually performs corresponding measurements on fresh EPR pairs, and leaves the state $\ket{\psi}$ unchanged.
Thus $\prover_O$ on receiving the teleported quantum state, can run a swap test to check if the returned state has been changed.

To overcome the issue, our approach is to run the tests in the Leash protocol only once: 
in the beginning of the protocol, the verifier tosses a biased coin $\gamma$ such that $\Pr[\gamma=0]=O(1/q)$. If $\gamma=1$, the verifier chooses a uniformly random index $\ell\in\{1,\ldots,q\}$, and performs a random test. %
In particular, when $\gamma=1$, the verifier chooses a random iteration $i\in[q]$, and runs the Leash protocol for the $i^{th}$ query, i.e., running a $X$-test, a $Z$-test, or a rigidity test with probability determined later.
If the provers pass the test, then the verifier simply accepts and terminates (i.e., ignores the rest of the queries).
If $\gamma=0$, the verifier only performs the same check as in the original query algorithm, i.e., it checks whether $\prover_A$'s final output $w$ satisfies $R(O,w)=1$.

We describe the protocol in the following steps.
Note that though the provers do not necessarily follow the steps, we include the honest behavior for concreteness and clarity.
\begin{protocol}[$\query(q,\O,V)$]\label{prot:query}~
Let $p=1/3$ and $\alpha=\Theta(1/q)$.
\begin{enumerate}
\item The verifier $\verifier$ samples a oracle $O\gets_R\O$ for oracle ensemble $\O$ and a bit $\gamma$ with $\Pr[\gamma=0]=\alpha$. %
  If $\gamma=1$, then $\verifier$ samples a random index $\ell\gets_R[q]$.

\item If $\gamma=1$, for $i\in\{1,\ldots,\ell-1\}$, run $\comp(\Sigma,O,m)$.
  For $i=\ell$, 
  \begin{enumerate}
  \item $\verifier$ chooses a random protocol from $\{\xtest(\Sigma,O,m),\ztest(\Sigma,O,m),\rigid(\Sigma,m)\}$ with probability $p_X=p,p_Z=p,p_R=1-2p$ respectively. %
  \item If the test succeeds then the verifier accepts and terminates; otherwise it rejects.
  \end{enumerate}

  Otherwise, $\gamma=0$, for $i\in\{1,\ldots,q\}$, run $\comp(\Sigma,O,m)$.
  Finally $\prover_A$ sends an answer $w$, and the verifier outputs 1 if $V(O,w)=1$ and 0 otherwise.
\end{enumerate}
\end{protocol}

The protocol $\query$ (\prot{query}) describes $\verifier$'s behavior to initiate the protocol. 
More specifically, $\verifier$ randomly decides to do computation or tests first. 
If $\verifier$ decides to do tests, it randomly picks one iteration $\ell$ from $[q]$ and randomly selects a test ($X$, $Z$, or rigidity) with probability $p_X,p_Z,p_R$ respectively. 
Then $\verifier$ applies corresponding protocols $\xtest$ (\prot{X}), $\ztest$ (\prot{Z}), and $\rigid$ (\prot{R}) at iteration $\ell$ and always applies $\comp$ (\prot{comp}) to other rounds.  

Steps 1 to 4 in these four protocols are the same. 
In particular, the verifier $\verifier$ divides the EPR pairs shared between $\prover_O$ and $\prover_A$ into subsets $N_C$, $N_X$, and $N_Z$ for running $\comp$ (\prot{comp}), $\xtest$ (\prot{X}), and $\ztest$ (\prot{Z}) with corresponding measurements (chosen from the set $\Sigma$ for each EPR pair). 
The registers $B_1,\dots,B_d$ are used for applying the $T$ gadgets. 
Then, the verifier asks $\prover_A$ to measure its halves of the EPR pairs according to $W$, to teleport his queries to $\prover_O$, and to send the corresponding Pauli correction to the $\verifier$. 
The verifier $\verifier$ sets $N$ according to the round type and reveals $N$ only to $\prover_O$, but the underlying state remains unknown to $P_O$ since it is encrypted by a quantum one-time pad. 
On the other hand, $\prover_A$ cannot distinguish $\comp$ (\prot{comp}), $\xtest$ (\prot{X}) or $\ztest$ (\prot{Z}) since it does not know $N$. Moreover, from $\prover_A$'s view, $\rigid$ (\prot{R}) looks the same as the other tests. This implies that $P_A$ cannot know if $\verifier$ is chooses to run $\rigid$ (\prot{R}) to check its behavior. 

\begin{protocol}[$\comp(\Sigma,O,m)$]\label{prot:comp}~
  \begin{enumerate}

    \item Setup:
      \begin{enumerate}
      \item $\verifier$ samples a random subset $N_C\subseteq[m]$ of size $n$, where $n$ is the number of qubits $O$ acts on. %
      \item For each $i\in\bar N_C$, $\verifier$ samples $W_i\gets_R\Sigma$.
      \item Then $\verifier$ samples two random subsets $N_X\subseteq\{i: W_i=X\}$ and $N_Z\subseteq\{i: W_i=Z\}$, both of size $n$.
      \item $\verifier$ sets $N=N_C$.

      \item $\verifier$ partitions $[m]\backslash (N_C\cup N_X\cup N_Z)$ into subsets of equal size $B_1,\ldots,B_d$, where $d$ is the number of $T$ layers in $O$. 
      \end{enumerate}

    \item $(\verifier\to\prover_A)$ $\verifier$ sends $N_C, W_{\bar N_C}$ to $\prover_A$ sequentially.
    \item $(\prover_A\to\verifier)$ 
      $\prover_A$ teleports its quantum state $\ket{\psi}$ using EPR pairs with indices in $N$, and sends the Pauli correction $(a_{N_C},b_{N_C})$.
      For each $i\in\bar N_C$, perform measurements on the $i^{th}$ EPR pair in basis $W_i$.
      Finally $\prover_A$ sends the outcomes $e_{\bar N_C}$.
  \item $(\verifier\to\prover_O)$: $\verifier$ sends $N$ to $\prover_O$.

  \item For each $\ell=1,\ldots,d$, 
    \begin{enumerate}
    \item $(\verifier\to\prover_O)$ $\verifier$ chooses a random subset $T_\ell\subseteq \{i\in B_\ell: W_i\in\{G,F\}\}$, and sends $T_\ell$ to $\prover_O$.
      For each Clifford gate in the $\ell$-th layer, perform the appropriate key update.

    \item $(\prover_O\to\verifier)$: $\prover$ performs the Clifford operations in the $\ell$-th layer. For each $T$-gadget in the $\ell$-th layer, $\prover_O$ runs the $T$-gadget on $(i,j)\in T_\ell\times N$, and sends the measurement outcome $c_{T_\ell}$ to $\verifier$.
    \item $(\verifier\to\prover_O)$ For each $i\in T_\ell$, set $z_{i}=a_{j}+1_{W_i=F}+c_i$. $\verifier$ sends $z_{T_\ell}$ to $\prover_O$.
    \end{enumerate}

  \item $(\prover_O\to\verifier)$ Let $\ket{\phi}$ be the resulting state after the evaluation of $O$ is done. 
    $\prover_O$ teleports $\ket{\phi}$ to $\prover_A$ and sends the Pauli correction $a',b'$ to $\verifier$.
  \end{enumerate}
\end{protocol}

In the following steps in \prot{comp}, $\verifier$ basically guides $\prover_O$ to apply the oracle on $\ket{\psi}$ using corresponding gadgets and updates the keys according to the measurement outcomes of $\prover_O$ and $\prover_A$. 

\begin{protocol}[$\xtest(\Sigma,O,m)$]\label{prot:X}~
  \begin{enumerate}
    \item Setup:
      \begin{enumerate}
      \item $\verifier$ samples a random subset $N_C\subseteq[m]$ of size $n$.
      \item For each $i\in\bar N_C$, $\verifier$ samples $W_i\gets_R\Sigma$.
      \item Then $\verifier$ samples two random subsets $N_X\subseteq\{i: W_i=X\}$ and $N_Z\subseteq\{i: W_i=Z\}$, both of size $n$.
      \item $\verifier$ sets $N=N_X$.

      \item $\verifier$ partitions $[m]\backslash (N_C\cup N_X\cup N_Z)$ into subsets of equal size $B_1,\ldots,B_d$, where $d$ is the number of $T$ layers in $O$.
      \end{enumerate}

    \item $(\verifier\to\prover_A)$ $\verifier$ sends $N_C, W_{\bar N_C}$ to $\prover_A$ sequentially.
    
    \item $(\prover_A\to\verifier)$ 
      $\prover_A$ teleports its quantum state $\ket{\psi}$ using EPR pairs with indices in $N$, and sends the Pauli correction $(a_{N_C},b_{N_C})$.
      For each $i\in\bar N_C$, perform measurements on the $i^{th}$ EPR pair in basis $W_i$.
      Finally $\prover_A$ sends the outcomes $e_{\bar N_C}$.
  \item $(\verifier\to\prover_O)$: $\verifier$ sends $N$ to $\prover_O$.
  \item For each $\ell=1,\ldots,d$, 
    \begin{enumerate}
    \item $(\verifier\to\prover_O)$ $\verifier$ chooses a random subset $T_\ell=T_\ell^0\cup T_\ell^1$ such that $T_\ell^0$ is a random subset of $\{i: W_i=Z\}$, and $T_{\ell}^1$ is a random subset of $\{i:W_i\in\{X,Y\}\}$. $\verifier$ sends $T_\ell$ to $\prover_O$.
      For each Clifford gate in the $\ell$-th layer, perform the appropriate key update.

    \item $(\prover_O\to\verifier)$: $\prover_O$ performs the Clifford operations in the $\ell$-th layer. For each $T$-gadget in the $\ell$-th layer, $\prover_O$ runs the $T$-gadget on $(i,j)\in T_\ell\times N$, and sends the measurement outcome $c_{T_\ell}$ to $\verifier$.
    \item $(\verifier\to\prover_O)$ For each $i\in T_\ell^0$, set $z_{i}\gets_R\bit$; if $i\in T_{\ell}^1$, $z_i=1_{W_i=Y}$. $\verifier$ sends $z_{T_\ell}$ to $\prover_O$.
    $\prover_O$ teleports $\ket{\phi}$ to $\prover_A$ and sends the Pauli correction $a',b'$ to $\verifier$.
    \end{enumerate}
  \item $(\prover_O\to\verifier)$ Let $\ket{\phi}$ be the resulting state after the evaluation of $O$ is done. 
    $\prover_O$ teleports $\ket{\phi}$ to $\prover_A$ and sends the Pauli correction $a',b'$ to $\verifier$.
  \item ($\verifier\to\prover_A$) $\verifier$ requests $\prover_A$ to perform a standard basis measurement on every qubit of the teleported state.
  \item ($\prover_A\to\verifier$) $\prover_A$ performs a standard basis measurement on every qubit of the teleported state and sends the outcome $d$.
    $\verifier$ accepts if and only if $d\oplus a' \oplus a''=0$, where $a''$ is calculated by the key update rule.
  \end{enumerate}
\end{protocol}

\begin{protocol}[$\ztest(\Sigma,O,m)$]\label{prot:Z}~
  \begin{enumerate}
    \item Setup:
      \begin{enumerate}
      \item $\verifier$ samples a random subset $N_C\subseteq[m]$ of size $n$.
      \item For each $i\in\bar N_C$, $\verifier$ samples $W_i\gets_R\Sigma$.
      \item Then $\verifier$ samples two random subsets $N_X\subseteq\{i: W_i=X\}$ and $N_Z\subseteq\{i: W_i=Z\}$, both of size $n$.
      \item $\verifier$ sets $N=N_Z$.

      \item $\verifier$ partitions $[m]\backslash (N_C\cup N_X\cup N_Z)$ into subsets of equal size $B_1,\ldots,B_d$, where $d$ is the number of $T$ layers in $O$.
      \end{enumerate}

    \item $(\verifier\to\prover_A)$ $\verifier$ sends $N_C, W_{\bar N_C}$ to $\prover_A$ sequentially.
    \item $(\prover_A\to\verifier)$ 
      $\prover_A$ teleports its quantum state $\ket{\psi}$ using EPR pairs with indices in $N$, and sends the Pauli correction $(a_{N_C},b_{N_C})$.
      For each $i\in\bar N_C$, perform measurements on the $i^{th}$ EPR pair in basis $W_i$.
      Finally $\prover_A$ sends the outcomes $e_{\bar N_C}$.
  \item $(\verifier\to\prover_O)$: $\verifier$ sends $N$ to $\prover_O$.
  \item For each $\ell=1,\ldots,d$, 
    \begin{enumerate}
    \item $(\verifier\to\prover_O)$ $\verifier$ chooses a random subset $T_\ell=T_\ell^0\cup T_\ell^1$ such that $T_\ell^0$ is a random subset of $\{i: W_i=\{X,Y\}\}$, and $T_{\ell}^1$ is a random subset of $\{i:W_i\in Z\}$. $\verifier$ sends $T_\ell$ to $\prover_O$.
      For each Clifford gate in the $\ell$-th layer, perform the appropriate key update.

    \item $(\prover_O\to\verifier)$: $\prover_O$ performs the Clifford operations in the $\ell$-th layer. For each $T$-gadget in the $\ell$-th layer, $\prover_O$ runs the $T$-gadget on $(i,j)\in T_\ell\times N$, and sends the measurement outcome $c_{T_\ell}$ to $\verifier$.
    \item $(\verifier\to\prover_O)$ For each $i\in T_\ell^1$, set $z_{i}\gets_R\bit$; if $i\in T_{\ell}^0$, $z_i=1_{W_i=Y}$. $\verifier$ sends $z_{T_\ell}$ to $\prover_O$.
    $\prover_O$ teleports $\ket{\phi}$ to $\prover_A$ and sends the Pauli correction $a',b'$ to $\verifier$.
    \end{enumerate}
  \item $(\prover_O\to\verifier)$ Let $\ket{\phi}$ be the resulting state after the evaluation of $O$ is done. 
    $\prover_O$ teleports $\ket{\phi}$ to $\prover_A$ and sends the Pauli correction $a',b'$ to $\verifier$.
  \item ($\verifier\to\prover_A$) $\verifier$ requests $\prover_A$ to perform a Hadamard basis measurement on every qubit of the teleported state.
  \item ($\prover_A\to\verifier$) $\prover_A$ performs a Hadamard basis measurement on every qubit of the teleported state and sends the outcome $d$.
    $\verifier$ accepts if and only if $d\oplus b' \oplus b''=0$, where $b''$ is calculated by the key update rule.
  \end{enumerate}
\end{protocol}

\prot{X} and \prot{Z} originate from the Broadbent protocol to check if $\prover_O$ is consistent with $O$. 
The crucial idea, as introduced in \sec{broadbent}, is that $\prover_O$ acts as applying identity on an all-zero or an all-plus state up to key update. 
Therefore, the verifier can detect if there is an attack applied to the computation: First $V$ asks $P_A$ to perform a standard basis or a Hadamard basis measurement on the state received from $P_O$.
Then it applies the key update rules to compute the decryption key, and check if the measurement outcomes from $P_A$ decrypts to zero.

Since the above steps relies on reliable measurements performed by $P_A$, it is essential to enforce $P_A$ to perform the measurements correctly.
We include our modification $\rigid(\Sigma,m)$ of the rigidity test in \sec{leash}.
The test is the same as $\rigid'(\Sigma,|\bar N|)$ on a random subset $\bar N$ of $[m]$. %
The purpose of $\rigid$ (\prot{R}) is to check if $\prover_A$ measures its EPR pairs in bases $W'$. 
From the collection of measurement outcomes for questions $W$ to $P_A$ and $W'$ to $P_O$, $V$ checks if the outcomes follows the relation specified in \prot{R}.
By \thm{cliff}, passing with probability $1-\epsilon$ ensures $P_A$'s output is $\poly(\epsilon)$-close in total variation distance to a measurement performed on EPR pairs in the correct bases $W$.
Note that $\prover_O$ knows that $\verifier$ chooses to execute $\rigid$ (\prot{R}) after receiving $W$ from $\verifier$. However, $\prover_A$ does not know this since it only receives random partitions and $W$, indistinguishable from the messages he obtained in other protocols. 
\begin{protocol}[$\rigid(\Sigma,m)$]\label{prot:R}~
  \begin{enumerate}
      \item Setup: 
        \begin{enumerate}
            \item $\verifier$ samples a random subset $N_C\subseteq[m]$ of size $n$.
            \item For each $i\in \bar N_C$, $\verifier$ samples $W_i\gets_R\Sigma$.
            \item Then $\verifier$ samples two random subsets $N_X\subset\{i:W_i=X\}$ and $N_Z\subseteq\{i:W_i=Z\}$, both of size $n$.
            \item $\verifier$ sets $N=N_C$.
            \item $\verifier$ partitions $[m]\backslash(N_C\cup N_X\cup N_Z)$ into subsets of equal size $B_1,\ldots,B_d$, where $d$ is the number of $T$ layers in $O$.
        \end{enumerate}

    \item Execute $\rigid'(\Sigma,|\bar N|)$ on the subset $\bar N$ and output the outcome.
  \end{enumerate}
\end{protocol}

\subsection{Completeness}

The completeness immediately follows: for algorithm $\A$, the provers $\prover_O$ performs $O$ and $\prover_A$ performs $U_i$ in each iteration $i$.
Then the provers with probability 1 if $\xtest$ or $\ztest$ is chosen.
By \thm{cliff}, when $\rigid$ is chosen, they succeeds with probability $1-\exp(-\Omega(n+t))$. 
We give a proof as follows.
\begin{theorem}[Completeness]\label{thm:leash-completeness}
  For every $q$-query algorithm $\A^O$ which outputs $w$ satisfying $R(O,w)=1$ with probability at least $2/3$, there exist provers $\prover_O,\prover_A$ which succeed with probability at least $1-\frac{\alpha}{3}-\exp(-\Omega(n+t))$.
\end{theorem}
\begin{proof}
For every algorithm $\A$ that succeeds with probability at least $c'\geq 2/3$, in each iteration $i$, $\prover_O$ runs $O$ and $\prover_A$ runs $U_i$.
The provers passes $\xtest$ and $\ztest$, when they are chosen, with probability 1.
By \thm{cliff}, when the verifier chooses to execute $\rigid$ (with probability $(1-\alpha)/3$), honest provers succeed with probability $1-\exp(-\Omega(n+t))$ since CHSH games are run in sequential repetition.
Thus the success probability of the provers is
\begin{align}\nonumber
  \alpha\cdot c' + \frac{2(1-\alpha)}{3} + \frac{1-\alpha}{3} (1-\exp(-\Omega(n+t)))
  &= 1-\alpha(1-c') - \frac{1-\alpha}{3}\exp(-\Omega(n+t)) \\
  &\geq 1-\frac{\alpha}{3} - \exp(-\Omega(n+t)).
\end{align}
  
\end{proof}

\subsection{Soundness}

To show the soundness, recall that it suffices to show that if the provers succeed with probability more than $s$, then there exists a query algorithm which is accepted with probability more than 1/3.
First we show a simpler case in which $\prover_A$ behaves honestly.

\begin{lemma}[Soundness with honest $\prover_A$]\label{lem:soundness-honest-pa}
  There exists $\alpha=\Theta(1/q)$ such that, for every provers $\prover_O$ that succeeds with success probability $s>1-\frac{2\alpha}{3}$, there exists a query algorithm that is given access to $O$ and outputs $w$ satisfying $R(O,w)=1$ with probability more than $1/3$. 
\end{lemma}
\begin{proof}
  Let the success probability be $1-\epsilon_i$ conditioned on $\gamma=1$ and the chosen index is $i$.
  The the probability of failing $\xtest$ or $\ztest$ is at most $p^{-1}\epsilon_i$.
  Then the quantum channel $\E_i$ that $\prover_A$ implements at iteration $i$ satisfies $\|\E_i-\O\|_\diamond\leq 2p^{-1}\epsilon_i$.

  Then let $\A$ be $\U_q\circ\O\circ\U_{q-1}\circ\cdots \circ \U_1\circ\O\circ \U_0$.
  By union bound, 
  \begin{align}\label{eq:distance-soundness}
    \|\A - \U_q\circ\E_q\circ\U_{q-1}\circ\cdots \circ \U_1\circ\E_1\circ \U_0\|_{\diamond} \leq \frac{2}{p}\sum_{i=1}^t \epsilon_i = 2qp^{-1}\epsilon,
  \end{align}
  where $\epsilon:=\frac{1}{q}\sum_i\epsilon_i \in [0,1]$.
  Then the success probability of $\prover_O$ is 
  \begin{align}
  s= (1-\alpha)(1-\epsilon)+\alpha\delta
  \leq(1-\alpha)(1-\epsilon)+\alpha (p_\A + 2q\epsilon/p),
    \end{align}
  where $\delta$ is the probability that the provers output $w$ such that such that $R(O,w)=1$ conditioned on $\gamma=0$ (i.e., the second quantum channel in \eq{distance-soundness}), and $p_\A$ is the success probability of $\A$.
  This implies that 
  \begin{align}
    p_\A \geq \frac{s}{\alpha} - \Big(\frac{1}{\alpha}-1\Big) (1-\epsilon) - 2q\epsilon/p = 1 - \frac{1-s}{\alpha} + \epsilon \Big(\frac{1}{\alpha}-1-\frac{2q}{p}\Big).
  \end{align}
  Setting $\alpha = \frac{1}{1+2q\cdot c/p}$ for any constant $c>0$, we conclude that $p_\A\geq 1-\frac{1-s}{\alpha}$.
  If $s>1-\frac{2\alpha}{3}$, $p_\A> 1/3$.
\end{proof}

Now we consider the effect of a cheating $\prover_A$.
The crucial idea is if $P_A$ chooses to deviate non-trivially from the protocol by $\epsilon$ in total variation distance, then the probability it is accepted when $\rigid$ (\prot{R}) is chosen is then at most $1-\epsilon$.
As argued previously, since $P_A$ does not learn whether $\rigid$ is selected, the same strategy must have been applied for other tests.
This implies that in the delegation game (where $\xtest$, $\ztest$ or $\comp$ is chosen), the score can be at most increased by at most $\poly(\epsilon)$.
This is because as shown in \thm{cliff}, the rigidity test guarantees that for every pair of provers succeeds with probability $1-\epsilon$, the output transcript must be $\poly(\epsilon)$-close to the that from an honest strategy in total variance distance.
More formally, for every pair of provers $P_O$ and $P_A$ such that they succeed in the rigidity test with probability $1-\epsilon$ and in the delegation game with probability $q$,
there exist $P_O'$ and $P_A'$ such that $P_A'$ plays honestly (i.e., performs a correct measurement on a half of every EPR pairs) and they succeed in the delegation game with probability $q-\poly(\epsilon)$.
We use the result to prove the following theorem.
\begin{theorem}[Soundness]\label{thm:leash-soundness}
  For constant $p$, there exists $\alpha=1/\poly(q)$, such that for every pair of provers that succeeds with probability $s>1-\frac{2\alpha}{3}$, 
  there exists a query algorithm that is given access to $O$ and outputs $w$ satisfying $R(O,w)=1$ with probability more than $1/3$.
\end{theorem}
\begin{proof}
  Let the success probability be $1-\epsilon_i$ conditioned on $\gamma=1$ and the chosen index being $i\in[q]$.
  Thus the failure probabilities are at most $\frac{\epsilon_i}{p}$, $\frac{\epsilon_i}{p}$ and $\frac{\epsilon_i}{1-2p}$ respectively conditioned on the events that an $\xtest$, an $\ztest$ and an rigidity test $\rigid$ is chosen.
  Also note that when an $\xtest$ or a $\ztest$ is chosen, the provers do not distinguish the test from $\comp$ until $\verifier$ asks a measurement from $\prover_A$.
  When $\rigid$ is chosen, $\prover_A$ does not distinguish it from $\comp,\xtest,\ztest$ until $\verifier$ accepts or rejects.

  By \thm{cliff}, there exist $\prover_A',\prover_O'$ such that $\prover_A'$ plays honestly, and $\prover_O'$ successfully passes $\xtest$ and $\ztest$ with probability at least $1-\delta_i=1-\frac{\epsilon_i}{p}-\poly\big(\frac{\epsilon_i}{1-2p}\big)$.
  Thus by \thm{broadbent}, $\prover_O'$ implements a quantum channel $\E_i$ such that $\|\E_i-\O\|_{\diamond}\leq 1-2\delta_i$.

  Conditioned on $\gamma=0$, let the process of $\prover_A'$ on receiving a teleported state $\rho_{in}^{(i)}$, produces the output state $\rho_{out}^{(i)}$ be $\U_i:\rho_{in}^{(i)}\mapsto\rho_{out}^{(i)}$.
  Now let the algorithm $\A:= \U_q\circ\O_q\circ\cdots\circ\U_1\circ\O_1\circ\U_0$.
  By union bound,
  \begin{align}
    \|\A - \U_q\circ\E_q\circ\U_{q-1}\circ\cdots \circ \U_1\circ\E_1\circ \U_0\|_{\diamond} \leq 2\sum_{i=1}^q \delta_i, = 2q\delta,
  \end{align}
  where $\delta=\frac{1}{q}\sum_{i=1}^q\delta_i=\epsilon/p+\poly(\frac{\epsilon}{1-2p})\leq g(\epsilon)=\poly(\epsilon)$ for some monotonically increasing $g$ in $[0,\infty)$ (e.g., $c\cdot\epsilon^a$ for constants $a\leq 1$ and $c$).
  Since $g$ is monotonically increasing for $\epsilon\geq 0$, we note that the inverse $g^{-1}$ exists.
  The success probability of the provers is upper bounded by 
  \begin{align}
    s \leq (1-\alpha)(1-\epsilon) + \alpha\cdot \max\{ (p_\A+2q \cdot g(\epsilon)), 1\} 
  \end{align}
  where $\epsilon=\frac{1}{q}\sum_i\epsilon_i$ and $p_\A$ is the probability that measuring the associated qubits on $\A$'s output state yields an outcome $w$ satisfying $R(O,w)=1$.
  Since $g(\epsilon)$ is monotonically increasing, there exists $\epsilon^*\geq 0$ such that $2q\delta(\epsilon^*)=1-p_\A\leq 2q\cdot g(\epsilon^*)$.
  This implies that
  \begin{align}
      s\leq (1-\alpha)(1-\epsilon^*) + \alpha = 1-(1-\alpha)\epsilon^* 
      \leq 1-(1-\alpha)\cdot g^{-1}\Big(\frac{1-p_\A}{2q}\Big),
  \end{align}
  and
  \begin{align}\label{eq:pa-bound}
      p_\A \geq 1 - 2q\cdot  g\left(\frac{1-s}{1-\alpha}\right).
  \end{align}
By \eq{pa-bound}, if $s>1-(1-\alpha)\cdot g^{-1}(\frac{1}{3q})$, $p_\A>1/3$.
  For $\alpha > \frac{g^{-1}(\frac{1}{3q})}{\frac{2}{3}+g^{-1}(\frac{1}{3q})}$, we have $1-(1-\alpha)\cdot g^{-1}(\frac{1}{3q})> 1-\frac{2\alpha}{3}$.
  It suffices to choose $\alpha=\frac{2\cdot g^{-1}(\frac{1}{3q})}{\frac{2}{3}+g^{-1}(\frac{1}{3q})}=1/\poly(q)$.
\end{proof}

Setting $p=1/3$, we conclude with the following corollary, a direct consequence of \thm{leash-completeness} and \thm{leash-soundness}.
\begin{corollary}\label{cor:main}
  Let $\C_{yes},\C_{no}$ be two complexity classes. If there exists an oracle $\O$ and a relation $R$ such that $R$ is solvable in $\C_{yes}^\O$ using $q$ queries but not in $\C_{no}^\O$. Then, there exists $\alpha=1/\poly(q)$ and a protocol $\langle V, P_O, P_A\rangle$ such that the following statements hold.
  \begin{itemize}
      \item There exist $\prover_O$ that runs in $O(q\cdot CC(\O))+\poly(n)$ and $\prover_A$ runs in $\C_{yes}$ such that the verifier accepts with proability at least $1-\frac{\alpha}{3}-e^{-\Omega(n)}$.
      
      \item For any $\prover_A$ that runs in $\C_{no}$ and any unbounded $\prover_O$, the verifier accepts with probability at most $1-\frac{2\alpha}{3}$.
  \end{itemize}
  Here, $CC(\O)$ is the quantum circuit complexity for implementing $\O$. 
\end{corollary}

\section{Application: classical verification of quantum depth}
\label{sec:CVQD}

In this section, we will prove the existence of $\CVQD_2(d,d')$ for integers $d'>d$. First we give the formal definition as follows. 
\begin{definition}[$\CVQD_2(d,d')$]
\label{dfn:cvqd_2p_formal}
Let $d,d'\in \mathbb{N}$ and $d'>d$. 
A two-prover protocol $\CVQD_2(d,d')$  that separates quantum circuit depth $d$ from $d'$ consists of a classical verifier $V$ and two provers $P_O,P_A$ such that the following properties hold:
\begin{itemize}
    \item \textup{\bf Non-locality:} $P_O$ and $P_A$ share an arbitrary quantum state, and there is no quantum and classical channel between them once the protocol starts. 
    \item \textup{\bf Completeness:} %
    There exist an integer $\hat d\geq d'$, a quantum prover $P_O$ and a $\hat{d}$-QC or $\hat{d}$-CQ scheme $P_A$ such that $\Pr[\langle V,P_O,P_A\rangle=accept]\geq 2/3$.
    \item \textup{\bf Soundness:} 
    For integer $\hat d\leq d$ and any $\hat{d}$-QC or $\hat{d}$-CQ scheme $P_A$ and any $P_O$, $\Pr[\langle V,P_O,P_A\rangle=accept]\leq 1/3$. 
\end{itemize}%
\end{definition}

We prove the following theorem by \cor{main} and the oracle separation in~\cite{CCL19}. Note that  \dfn{cvqd_2p_formal} does not specify the power of $V$ and honest $P_O$ except that $V$ is a classical algorithm and $P_O$ is a quantum machine that can store EPR pairs. we first show the existence of an information-theoretically secure $\CVQD_2(d,2d+3)$ for any $d=\poly(n)$ with inefficient verification, i.e., $V$ and honest $P_O$ need to run in exponential time. Then, we show a $\CVQD_2(d,2d+3)$ for any $d=\poly(n)$ with efficient verification under the assumption that qPRP (\dfn{qprp}) exists.  

\begin{theorem}
\label{thm:cvqd_2p_formal}
Let $d = \poly(n)$. 
\begin{enumerate}
    \item There exist $\alpha=1/\poly(d)$ and an unconditionally
    secure $\CVQD_2(d,2d+3)$ (\dfn{cvqd_2p_formal}), in which the verifier $V$ runs in probabilistic $O(2^n)$ time such that the following holds.
    \label{itm:cvqd_2p_1}
      \begin{itemize}
        \item \textup{\bf Completeness:} There exist $P_A$ which has quantum depth at least $2d+3$ and $P_O$ which runs in quantum $O(2^n)$ time such that $\Pr[\langle V,P_O,P_A\rangle=accept]\geq 1-\frac{\alpha}{3}$.

        \item \textup{\bf Soundness:} For any unbounded $\prover_O$ and $\prover_A$ that are $d$-CQ or $d$-QC schemes, $\Pr[\langle V,\prover_O, \prover_A\rangle = accept]\leq 1-\frac{2\alpha}{3}$.
      \end{itemize}
    \item Assume that there exist quantum-secure pseudorandom permutations (qPRP) as defined in \dfn{qprp}.
      There exist $\alpha=1/\poly(d)$ and $\CVQD_2(d,2d+3)$ (\dfn{cvqd_2p_formal}), in which  $V$ runs in probabilistic polynomial time such that the following holds. %
        \label{itm:cvqd_2p_2}
    \begin{itemize}
        \item \textup{\bf Completeness:} There exist $P_A$ that has quantum depth at least $2d+3$ and $P_O$ that runs in quantum polynomial time such that $\Pr[\langle V,\prover_O, \prover_A\rangle = accept]\geq 1-\frac{\alpha}{3}$. %
        \item \textup{\bf Soundness:} For any unbounded $\prover_O$ and $\prover_A$ that are $d$-CQ or $d$-QC schemes, $\Pr[\langle V,\prover_O, \prover_A\rangle = accept]\leq 1-\frac{2\alpha}{3}$.
  \end{itemize}
\end{enumerate}
\end{theorem}

For efficient instantiations, it is required that the functions $f_0,\ldots,f_d$ are efficiently samplable and computable. Any construction of qPRP satisfying \dfn{qprp} (e.g., \cite{Zha16}) can be used to construct a pseudorandom $d$-shuffling of a pseudorandom Simon's function. We now give constructions.

In the problem $\dSSP$, the functions $f_0,\ldots, f_{d-1}$ are random permutations. 
These functions can be replaced with pseudorandom permutations.
For the last function $f_d$, we note that $f_d$ can be written as $f\circ f_{0}^{-1}\circ\cdots\circ f_{d-1}^{-1}$, where $f$ is a random Simon's function, when the domain is restricted to a hidden subset. 
It then suffices to show a construction of a pseudorandom Simon's function.

\begin{definition}[Pseudorandom Simon's function]\label{dfn:p-simon}
  For finite set $\Y$, let $\mathcal S$ be the set of Simon's function from $\bit^n$ to $\Y$, i.e., $f\in\mathcal S$ if there exists $s\in\bit^n$ such that $f(x)=f(x')$ if and only if $x=x'\oplus s$.
  For key space $\K$, a pseudorandom Simon's function is a keyed function $F:\K\times\bit^n\to\bit^{m}$ such that for every quantum adversary $\A$, it holds that
  \begin{align}
    \left|\Pr_{F\gets_R \mathcal S}[\A^F()=1] - \Pr_{k\gets_R\K}[\A^{F_k}()=1]\right| \leq \negl(n).
  \end{align}
\end{definition}
We note that by the definition of Simon's function, it must be the case that $m\geq n-1$.
Next we prove that there exists a pseudorandom Simon's function from qPRP.
\begin{claim}\label{clm:p-simon}
  Assume that qPRP exists as defined in \dfn{qprp}.
  Then there exists a pseudorandom Simon's function as defined in \dfn{p-simon}.
\end{claim}
\begin{proof}
  We first show that a random Simon's function can be constructed from a random permutation, and thus replacing a random permutation with a qPRP, we obtain a pseudorandom Simon's function.

Let $H:=\{x\in\bit^n: x<x\oplus s\}$ for total ordering $<$ over $\bit^n$ defined as follows:
For $x,y\in\bit^n$, $x<y$ if the smallest index $i\in[n]$ where $x_i\neq y_i$ satisfies $x_i=0$ and $y_i=1$.

The subset $H$ forms a subgroup of $\bit^n$ for group operation $\oplus$: It is clear that $0\in H$ since 0 is smaller than any string in $\bit^n$. 
Let $i\in[n]$ be the smallest index such that $s_i=1$.
For $x,y\in H$, $x_i=y_i=0$, and thus $(x\oplus y)_i=0$. 
This implies that $x\oplus y\in H$.
Since $H$ is a subgroup, the cosets $\{H,s\oplus H\}$ forms a partition of $\bit^n$.

Now we show that for codomain $\Y=\bit^m$ where $m\geq n-1$, every Simon's function $f:\bit^n\to\Y$ can be constructed from a permutation and a hidden shift $s$.
We then define the following function $T_s:\bit^n\to H$, $T_s(x)=x$ for $x\in H$, and $T_s(x)=x\oplus s$ for $x\in s\oplus H$. 
Let the mapping $W_s:H\to\bit^{m}$, 
\begin{align}
  W_s(x_1,\ldots,x_n)=(x_1,\ldots,x_{i-1},x_{i+1},\ldots,x_n, 0,\ldots,0),
\end{align}
where $i$ is the smallest index such that $s_i=1$.
The padding has $m-n+1$ zeros.
For every Simon's function $f$ with shift $s$, we know that $f=f\circ T_s$.
When the domain is restricted to $H$, $f$ is injective, and thus there exist $(2^m-|H|)!$ permutations $P:\bit^m\to\bit^m$ such that $f=P\circ W_s\circ T_s$ (we use the convention that $0!=1$).
On the other hand, by definition, for every $P$, $P\circ W_s\circ T_s$ is a Simon's function.
These facts imply that $(P,s)\mapsto P\circ W_s\circ T_s$ is a well-defined mapping from a pair of permutation and shift to a Simon's function and is $(2^m-|H|)!$-to-1.

Thus a random Simon's function can be sampled (inefficiently) as follows: 
First sample a uniform shift $s\gets_R\bit^n$ and a random permutation $P:\bit^m\to\bit^m$ and output $P\circ W_s\circ T_s$.
A pseudorandom Simon's function can be sampled efficiently using a qPRP $F:\K\times\bit^m\to\bit^m$: Sample a random shift $s\gets_R\bit^n$ and $k\gets_R\K$ and output $g_{k,s}:=F_k\circ W_s\circ T_s$.

To show that $g_{k,s}$ indeed yields a pseudorandom Simon's function, suppose toward contradiction there exists an adversary $\A$ which distinguishes $g_{k,s}$ for uniform $k,s$ from a random Simon's function with non-negligible probability $\epsilon$.
Then the reduction $\B$ given oracle access to a permutation $Q$, samples $s\gets_R\bit^n$ and outputs $b\gets\A^{Q\circ W_s\circ T_s}$.
If $Q$ is random, $Q\circ W_s\circ T_s$ is a random Simon's function; otherwise let the key be $k$, and the Simon's function is $g_{k,s}$. By the assumption, $\B$ distinguishes a random permutation from a qPRP with non-negligible advantage.
\end{proof}

We then define the pseudorandom $d$-shuffling of a function $f$, which can be implemented in quantum polynomial time. 
\begin{definition}[Pseudorandom $d$-shuffling (cf. \dfn{d-shuffling})]\label{dfn:p-d-shuffling}
  For $f:\bit^n\to\bit^n$, pseudorandom $d$-shuffling of $f$ is a tuple of functions $(f_0,\ldots,f_d)$, where $f_0,\ldots,f_{d-1}$ are pseudorandom permutations over $\bit^{(d+2)n}$, and the last function $f_d$ is a fixed function satisfying the following properties:
  Let $S_d:=\{f_{d-1}\circ \cdots \circ f_0(x'): x'\in\bit^n\}$.
  \begin{itemize}
  \item For $x\in S_d$, let $f_{d-1}\circ\cdots\circ f_0(x')=x$, and choose the function $f_d:S_d\to[0,2^{n}-1]$ such that $f_d\circ f_{d-1}\circ\cdots f_0(x')=f(x')$.

  \item For $x\notin S_d$, $f_d(x)=\bot$.
  \end{itemize}
\end{definition}

Now we are ready to define a pseudorandom $d$-shuffling Simon's problem.

\begin{problem}[Pseudorandom $\dSSP$ (cf.~\prob{dSSP})]\label{prob:p-dSSP}
  Given oracle access to a pseudorandom $d$-shuffling (\dfn{p-d-shuffling}) of a pesudorandom Simon's function (\dfn{p-simon}), the problem is to find the hidden shift.
\end{problem}
By a simple hybrid argument, pseudorandom $\dSSP$ separates a depth-$(2d+3)$ quantum computation from a depth-$d$ one.
\begin{corollary}
\label{cor:dssp_prp}
Let $d = \poly(n)$. Pseudorandom $\dSSP$ (\prob{p-dSSP}) can be solved by a $\QNC_{2d+3}$ circuit with classical post-processing. Furthermore, for any $\hat{d}$-CQ and $\hat{d}d$-QC schemes $\A$ with $\hat{d}\leq d$, the probability that $\A$ solves the problem is negligible. 
\end{corollary}

Now, we can prove \thm{cvqd_2p_formal}. 

\begin{proof} [Proof of {\thm{cvqd_2p_formal}}]
By \thm{dssp_ccl19}, $\dSSP$ separates the complexity classes $\BPP^{\BQNC_d}\cup \BQNC_d^{\BPP}$ and  $\BPP^{\BQNC_{2d+3}}\cap \BQNC_{2d+3}^{\BPP}$ relative to the $d$-shuffling oracle of $f$. Furthermore, teleporting quantum states only takes one circuit depth by choosing the gateset properly or considering the gateset including all one- and two-qubit gates as in~\rmk{gateset}.
Therefore, by setting $q=2d+1$, $\O$ to be the shuffling oracle, and $R$ to be $\dSSP$, \prot{query} separates depth-$(2d+3)$ quantum computation from depth-$d$ and the relation $R(O,w)=1$ if and only if $O$ is the shuffling oracle and $w$ is the hidden shift according (see \cor{main}). Here, $V$ and $\prover_O$ are inefficient since describing and implementing the shuffling oracle are inefficient. 

Following the same proof, we can also show that \prot{query} separates depth-$(2d+3)$ quantum computation from depth-$d$ by %
replacing $d$-SSP by pseudorandom $d$-SSP. This follows from the fact that pseudorandom $d$-SSP also separates depth-$(2d+3)$ quantum computation from depth-$d$ by \cor{dssp_prp}. Then, we follow the same proof for $d$-SSP by using \cor{main} except that we set $\O$ to be a pseudorandom-shuffling oracle that can be described and implemented efficiently, and thus both $V$ and $\prover_O$ are efficient. 
\end{proof}

Note that the algorithm is allowed to make parallel queries which do not increase the query depth. It is straightforward to adapt \prot{query} to allow parallel queries: let $t$ denote the largest number of parallel queries. For any query algorithm $\A$ of depth $q$, there is a query algorithm $\A'$ that is given access to $O^{\otimes t}$ and achieves the same performance as $\A$.
The equivalent two-player protocol to $\A'$ is $\query(2d+3,O', R)$, where sampling $O'$ can be performed by sampling $O\gets_R\O$ and outputting $O'=O^{\otimes t}$.

Furthermore, we emphasize that while the protocol only has a small completeness-soundness gap $\alpha/3=1/\poly(d)$, by sequential repetition for $O(\alpha^{-2}\cdot\log^2\lambda)$ times it suffices to amplify the gap to $1-\negl(\lambda)$.

\subsection{A nearly optimal separation}

In this section, we modify the original $\dSSP$ to give an oracle separation that reduces the gap from $d$ versus $2d+3$ to $d$ versus $d+3$. 

First, instead of considering standard quantum query model, we define $d$-shuffling in the ``in-place'' quantum oracle model.

\begin{definition}[In-place $d$-shuffling]
Let $f:\{0,1\}^n \rightarrow \{0,1\}^n$ be a Simon's function with shift $s$. Let $\F:=\{f_0,\dots,f_d\}$ be a $d$-shuffling of $f$.  We define the in-place $(d,f)$-shuffling $U:= \{U_{f_0},\dots,U_{f_d}\}$ as follows: 
\begin{enumerate}
    \item For $i=0$, let $U_{f_0}$ be a unitary such that for all $x\in \{0,1\}^{(d+2)n}$, $U_{f_0}\ket{x}\ket{0} = \ket{x}\ket{f_0(x)}$.
    \item For $i = 1,\dots,d-1$, let $U_{f_i}$ be a unitary in $\mathbb{C}^{2^{(d+2)n}\times 2^{(d+2)n}}$ such that for all $x\in \{0,1\}^{(d+2)n}$, $U_{f_i}\ket{x} = \ket{f_i(x)}$. 
    \item Let $U_{f_d}$ be a unitary in $\mathbb{C}^{2^{(d+2)n+1}\times 2^{(d+2)n+1}}$ such that for all $x\in \{0,1\}^{(d+2)n}$ and $b\in \{0,1\}$, $U_{f_d}\ket{x,b} = \ket{f_d(x)}\ket{b\oplus b'}$, where $b'=1$ if $x\in H$ (see the definition of $H$ in the proof of \clm{p-simon}).%
    \end{enumerate}
\end{definition}

We note that an in-place pseudorandom permutation exists if there exists a qPRP defined in \dfn{qprp}.
First we give the definition of in-place qPRPs.
\begin{definition}[In-place qPRP (cf.~\dfn{qprp})]\label{dfn:i-qprp}
  For security parameter $\lambda$ and polynomial $m=m(\lambda)$, a pseudorandom permutation $P$ over $\bit^m$ is a keyed function $\K\times\bit^m\to\bit^m$ such that there exists a negligible function such that for every quantum adversary $\A$, it holds that 
\begin{align}
    \Big|\Pr_{F\gets_R\P}[\A^{I_F,I_{F^{-1}}}=1]-\Pr_{k\gets_R\K}[\A^{I_{P(k,\cdot)}, I_{P^{-1}(k,\cdot)}}=1]\Big|\leq \negl(\lambda),
\end{align}
where $I_Q:\ket{x}\mapsto\ket{Q(x)}$ for $x\in\bit^m$ and permutation $Q:\bit^m\to\bit^m$.
\end{definition}
\begin{theorem}
  If there exists a qPRP as defined in \dfn{qprp}, then an in-place qPRP defined in \dfn{i-qprp} exists.
\end{theorem}
\begin{proof}
  It suffices to show that an in-place oracle of a permutation $P$ can be implemented using two queries to the standard oracles $O_P:\ket{x,y}\mapsto\ket{x,y\oplus P(x)}$ and $O_{P^{-1}}$.
  A query to the in-place oracle of $P$ can be computed in the following steps:
  \begin{align}\label{eq:i-qprp}
    \ket{x,0}\xmapsto{O_P} \ket{x,P(x)}
    \xmapsto{O_{P^{-1}}} \ket{x\oplus P^{-1}(P(x)),P(x)} = \ket{0,P(x)}.
  \end{align}
  Swapping the registers and removing the ancilla yields a mapping $\ket{x}\mapsto\ket{P(x)}$.

  Now we show that if $F:\K\times\bit^n\to\bit^n$ is a qPRP as defined in \dfn{qprp}, 
  then $F$ is an in-place qPRP.
  Suppose toward contradiction that a quantum adversary $\A$ which distinguishes $\ket{x}\mapsto\ket{F_k(x)}$ for $k\gets_R\K$ and $\ket{x}\mapsto\ket{P(x)}$ for uniform permutation $P$ with advantage $\epsilon$.
  Then we show there exists an adversary $\B$ which distinguishes $O_{F_k},O_{F_k^{-1}}$ from $O_P,O_{P^{-1}}$.
  The adversary $\B$ simulates the mapping in \eq{i-qprp} using two queries to either $O_{F_k},O_{F_k^{-1}}$ or $O_P,O_{P^{-1}}$ and runs $\A$.
  By the assumption, $\B$ distinguishes the oracles with advantage $\epsilon$.
\end{proof}
We define \textsl{in-place} $\dSSP$ as follows: 
\begin{definition}[In-place $\dSSP$]
\label{dfn:ipdssp}
Let $n\in\mathbb{N}$ and $f:\bit^n\to\bit^n$ %
be a random Simon's function.
Given access to the in-place $d$-shuffling oracle of $f$, the problem is to find the hidden shift of $f$.
\end{definition}

\begin{theorem}
\label{thm:cvqd_2p_opt}
  Let $d = \poly(n)$. In-place $\dSSP$ can be solved using $(d+3)$-CQ and $(d+3)$-QC schemes with access to the in-place $d$-shuffling oracle of $f$. Furthermore, for any $d'$-CQ or $d'$-QC schemes $\A$ with access to the in-place $d$-shuffling oracle of $f$ and $d'\leq d$, the probability that $\A$ solves the problem is negligible.
\end{theorem}

Note that if we consider the models defined in Definition 3.8 and Definition 3.10 in~\cite{CCL19}, the quantum depth separation will be $d$ versus $d+1$. See~\rmk{model} for the detailed discussion.   

We present the proof of \thm{cvqd_2p_opt} in \app{cvqd-2p-opt}. 
By \cor{main} and \thm{cvqd_2p_opt}, we have a construction of $\CVQD_2(d,d+3)$.
\begin{corollary}
\label{cor:cvqd_2p_final}
For $d=\poly(n)$, there exists an unconditionally secure $\CVQD_2(d,d+3)$ which is sound as in \thm{cvqd_2p_formal}, \itm{cvqd_2p_1} and complete with $P_O$ and $V$ running in $O(2^n)$ time. %
Moreover, if qPRP (\dfn{qprp}) exists, then there exists $\CVQD_2(d,d+3)$ which is sound as in \thm{cvqd_2p_formal}, \itm{cvqd_2p_2} and complete with $P_O$ and $V$ running in polynomial time. %
\end{corollary}
\begin{proof}%
Following the same idea, we set $q=d+1$, and we choose $\O$ to be the in-place shuffling oracle and $R$ to be in-place $\dSSP$. Then, \prot{query} separates $d+3$-depth quantum circuit from $d$-depth quantum circuit in the presence of polynomial-time classical computation by \cor{main} and \thm{cvqd_2p_opt}.

Again, $V$ and $\prover_O$ are not efficient since implementing random permutations is expensive. Following the same idea for proving the second result of \thm{cvqd_2p_formal}, we can make both $V$ and $\prover_O$ efficient using qPRP in the in-place oracle model. 
\end{proof}

\section{Certifying quantum depth from learning with errors}\label{sec:cqd-lwe}

In this section, we describe a protocol for certifying quantum depth using NTCFs.
First we define a single-prover protocol for certifying quantum depth.
\begin{definition}
\label{dfn:cvqd_2p_formal_single}
Let $d,d'\in \mathbb{N}$ and $d'>d$. 
A single-prover protocol $\CVQD(d,d')$ that separates quantum circuit depth $d$ from $d'$ consists of a classical verifier $V$ and a prover $P$ such that the following properties hold: %
\begin{itemize}
    
    \item \textup{\bf Completeness:} There exist an integer $\hat d\geq d$ and a $\hat d$-QC or $\hat d$-CQ scheme $P$ such that $\Pr[\langle V,P\rangle]\geq 2/3$.
    
    \item \textup{\bf Soundness:} For integer $\hat d\leq d$ and any $P$ that are $\hat d$-CQ or $\hat d$-QC schemes, %
    $\Pr[\langle V,P\rangle=accept]\leq 1/3$. 
\end{itemize}
\end{definition}

We give a prototocol based on a randomized encoding of NTCFs, and include the honest behavior of the prover in the description of the protocol (but the prover does not necessarily follow the instructions).

\begin{protocol}[{$\CVQD(d,d+d_0)$}]\label{prot:cqd-2}
  \begin{enumerate}
  \item The verifier samples $\{(k_i,t_i)\gets\Gen(1^\lambda): i\in[d+1]\}$ %
    and sends $(k_1,\ldots,k_{d+1})$. %

  \item %
    The prover performs the quantum process in \thm{lg21-completeness} and prepares the state
    \begin{align}\label{eq:claw}
      \frac{1}{(2|\X|)^{d/2}}\bigotimes_{i=1}^{d+1} \left(\sum_{b_i\in\bit,x_i\in\X}\ket{b_i}_{B_i}\ket{x_i}_{X_i}\ket{\hat f_{k_i}(b_i,x_i)}_{Y_i}\right),
    \end{align}
    and performs a standard basis measurement on the registers $Y_1,\ldots,Y_{d+1}$ to yield $y_1,\ldots,y_{d+1}$, which is sent to the verifier.
  \item %
    For $i=1\ldots d+1$, the verifier and the prover proceed as follows.
    \begin{enumerate}
    \item The verifier samples a random bit $c_i\in\bit$ and sends $c_i$ to the verifier. %
    
    \item %
      If $c_i=0$, the prover performs a standard basis measurement on $B_iX_i$; otherwise the prover performs a Hadamard basis measurement on $B_iX_i$.
      The prover then sends the outcome $w_i$ to the verifier.
      
    \item The verifier computes $V(t_i,y_i,c_i,w_i)=a$, where $V$ is defined in \eq{bcmvv-v}. If $a=0$, then the verifier rejects and aborts.
    \end{enumerate}
    If the verifier does not reject for each $i\in[d]$, then it accepts.
  \end{enumerate}
\end{protocol}

\begin{theorem}[Completeness]\label{thm:cqd-lwe-completeness}
  There exists a constant $d_0$ such that for security parameter $\lambda$ and polynomially bounded function $d$, a negligible function $\epsilon$, there is a prover which is a $(d_0+d(\lambda))$-QC scheme and succeeds with probability $1-\epsilon(\lambda)$.
\end{theorem}
\begin{proof}
  By \thm{lg21-completeness}, the preparation of the state in \eq{claw} can be done in constant depth $d'$. Let $d_0=d'+1$.
  The prover performs standard basis measurement on the registers $Y_1,\ldots,Y_{d+1}$ to sample $y_1,\ldots,y_{d+1}$.
  For every $i\in[d+1]$, the prover measures the state in the $i^{th}$ coordinate in the standard basis if $c_i=0$ and in the Hadamard basis otherwise.
  There exists a negligible function $\mu$ such that if $c_i=0$, with probability at least $1-\negl(\lambda)$, performing a standard basis on $B_iX_i$ yields a preimage;
  if $c_i=1$, with probability at least $1-\mu(\lambda)$, performing a Hadamard basis measurement on $B_iX_i$ yields an outcome that passes the equation test $V(t_i,y_i,1,\cdot)$.
  By the union bound, the prover succeeds with probability 
  \begin{align}\nonumber
      \Pr[\text{success}]
      &\geq 1-\sum_i\Pr[\text{Prover fails the $i^{th}$ round}] \\
      &\geq 1-d(\lambda)\cdot\mu(\lambda)=1-\negl(\lambda)
  \end{align}
  for polynomially bounded function $d$.
\end{proof}

As an example, we present the quantum circuit for $d=2$ in \fig{cqd-completeness}.

  \begin{figure}[htbp]
\begin{center}
\begin{minipage}{0.8\linewidth}
\vspace{.2em}
    \Qcircuit @C=1em @R=.7em {
\lstick{\ket{\psi_1}} & \meter & \gatew{h_1(\cdot)}\cw
 & \control\cw & \cw \\
\lstick{\ket{\psi_2}} & \qw & \qw & \gate{H}\cwx & \qw & \meter & \gatew{h_2(\cdot)}\cw & \control\cw & \cw \\
\lstick{\ket{\psi_3}} & \qw & \qw & \qw & \qw & \qw & \qw & \gate{H}\cwx & \qw & \meter & \gatew{h_3(\cdot)}\cw & \cw
}
\vspace{.2em}
\end{minipage}
\end{center}
\caption{The quantum circuit for $d=2$. In each layer $i\in[3]$, the verifier's action can be viewed as a classical computation $h_i$, which takes the measurement outcome as input. If the measurement outcome is accepted, then it outputs a random string $c_i\gets_R\bit^{\ell(\lambda)}$; otherwise it outputs $\bot$ indicating rejection.
Each $\ket{\psi_i}$ is the post-measurement state of \eq{claw} after a standard basis measurement on $Y_1,Y_2,Y_3$ is performed.
By \thm{lg21-completeness}, the state $\ket{\psi_1}\otimes\ket{\psi_2}\otimes\ket{\psi_3}$ can be prepared in constant depth.}
\label{fig:cqd-completeness}
\end{figure}
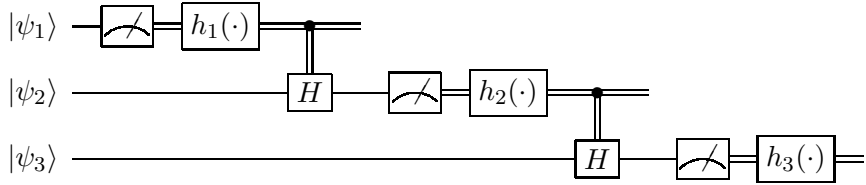

Next we show a lower bound on the quantum depth.
In each round $i$, let the prover's action on receiving challenge $c_i=c$ be an isometry $U_{i,c}$ acting on the quantum state $\ket{\psi_{i,T}}$, which depends on the previous transcript $T$, followed by a standard basis measurement.
We show that if there exists $i$ such that the quantum depth does not increase by 1 for round $i$, then there is a quantum adversary which breaks the adaptive hardcore bit property.

\begin{theorem}[Soundness]\label{thm:cqd-lwe-soundness}
  There exists a negligible function $\mu$ such that for sufficiently large $\lambda$, for any prover that is either a $d$-CQ or $d$-QC scheme succeeds with probability at most $\frac{3}{4}+\mu(\lambda)$.
\end{theorem}
\begin{proof}
  
  Suppose toward contradiction that there exists a prover $\prover$ which succeeds with probability $3/4+\epsilon$ for non-negligible $\epsilon$.
  First we show that the operation in each round must have non-zero quantum depth.
  If this is not the case for some round $i\in[d+1]$, then without loss of generality, $\prover$'s operation consists of
  \begin{enumerate}
      \item a standard basis measurement on some intermediate quantum state $\rho_{i}$ to yield an outcome $v_i$, followed by
      
      \item a classical algorithm $\A_i$ which on input $v_i$ and the challenge $c_i=c$, outputs the response $w_i\gets\A_i(c_i,v_i,T)$ for previous transcript $T$.
  \end{enumerate}
  Then there is a reduction $\A$ which uses $\prover$ to break the adaptive hardcore bit property.
  Since the probability of passing $d$ rounds is at least $3/4+\epsilon(\lambda)$, the probability of winning the first $i^{th}$ round is at least $3/4+\epsilon$.
  Thus $\A$, on input challenge key $k$, samples $k_1,\ldots,k_{i-1}, k_{i+1},\ldots,k_{d+1}$ and sets $k_i=k$ and  computes $v_i$.
  Using $v_i$, $\A$ runs $\A_i$ on $w_b\gets \A_i(b,v_i,T)$ for each $b\in\bit$.
  Let the probability that $w_b$ is a valid response be $p_b$.
  By the assumption, $p_0+p_1\geq 3/2+2\epsilon$.
  This implies that with probability $1-(1-p_0)-(1-p_1)=p_0+p_1-1\geq 1/2+2\epsilon$ \emph{both} $w_0$ and $w_1$ are valid.
  Thus with probability at least $1/2+2\epsilon$, $w_0$ is a valid preimage and $w_i$ is a valid equation, and thus the adaptive hardcore bit is broken.
  
  Since $\prover$ must have non-zero quantum depth in each round and it has total quantum depth $d$, 
  there must exist a round $j\in[d+1]$ such that $\prover$ must destroy all its coherence after receiving $c_j$, and continue answering the remaining rounds with an intermediate classical information $\sigma_j$.
  Now the reduction $\A'$ samples $c_1,\ldots,c_{d+1}\gets_R\bit$ and simulates the verifier in the protocol, and runs the following steps to break the adaptive hardcore bit property:
  \begin{enumerate}
      \item $\A'$, on receiving $k$, samples $k_1,\ldots,k_{d}$ and sets $k_{d+1}=k$, and runs $\prover$ to get $y_1,\ldots,y_{d+1}$ and some quantum information $\rho$.
  
      \item $\A'$ continues running $\prover$ on input $\rho$ and $c_1,\ldots,c_j$ to compute $\sigma_j$.
      
      \item $\A'$ continues running $\prover$ on input $\sigma_j,c_j,\ldots,c_{d},c_{d+1}=b$, to yield a response $w_b'$ in round $d$ for each $b\in\bit$.
      
      \item $\A'$ outputs $(y_{d+1},w_0',w_1')$ as the response.
  \end{enumerate}
  To analyze the performance of $\A'$, we apply the same idea as that for calculating the performance of $\A$.
  Let the probability that $w_b'$ be a valid response for round $d$ be $p_d'$.
  By the assumption, $p_0'+p_1'\geq 3/2+2\epsilon$.
  Thus both $w_0'$ and $w_1'$ are valid with probability at least $1-(1-p_0')-(1-p_1')\geq 1/2+2\epsilon$.
  This implies that the adaptive hardcore bit property is broken.
\end{proof}

\thm{cqd-lwe-completeness} and \thm{cqd-lwe-soundness} immediately imply the following theorem.
\begin{theorem}
  Assuming that LWE is hard for any $d$-CQ and $d$-QC schemes, \prot{cqd-2} satisfies the following conditions.
  \begin{itemize}
      \item \textup{\bf Completeness:} There exists a prover which is a  $(d_0+d)$-QC scheme and succeeds with probability at least $1-\negl(\lambda)$.
      
      \item \textup{\bf Soundness:}
      Every prover that are $d$-CQ and $d$-QC schemes succeed with probability at most $\frac{3}{4}+\negl(\lambda)$.
  \end{itemize}
\end{theorem}

By sequential repetition, the completeness-soundness gap can be amplified to $1-\negl(\lambda)$.
We note that we do not know if the function evaluation can be done by using $(d_0+d)$-CQ schemes and leave it as an open question.

\bibliographystyle{plain}
\bibliography{../refs}

\appendix

\section{Proof of \thm{cvqd_2p_opt}}\label{app:cvqd-2p-opt}

\begin{lemma}
In-place $\dSSP$ can be solved by a ($d+3$)-depth quantum circuit with classical post-processing. 
\end{lemma}
\begin{proof}
The algorithm is as follows: 
\begin{align*}
    \sum_{x\in \mathbb{Z}_2^n}\ket{x}\ket{0}\ket{0} &\xrightarrow{U_{f_0}} \sum_{x\in \mathbb{Z}_2^n}\ket{x}\ket{f_0(x)}\ket{0}\nonumber\\
    &\xrightarrow{U_{f_1}} \sum_{x\in \mathbb{Z}_2^n}\ket{x}\ket{f_1(x)}\ket{0}\nonumber\\
    &\xrightarrow{U_{f_d}} \sum_{x\in \mathbb{Z}_2^n}\ket{x}\ket{f(x)}\ket{b(x)}\nonumber\\
    &\xrightarrow{\mbox{measure the second register}} \frac{1}{\sqrt{2}}\left(\ket{x}\ket{b(x)}+\ket{x\oplus s}\ket{b(x\oplus s)}\right)\\
    &\xrightarrow{\mbox{Apply H on the last qubit and measure}} \frac{1}{\sqrt{2}}\left(\ket{x}+\ket{x\oplus s}\right) \mbox{ with probability } 1/2. \\
    &\xrightarrow{H^{\otimes n}} \frac{1}{\sqrt{2^n}}\sum_{j\in \mathbb{Z}_2^n} ((-1)^{x\cdot j} + (-1)^{(x+s)j}) \ket{j}
\end{align*}
Then, the rest of the algorithm follows from Simon's algorithm. The additional two depths come from the Hadamard gates an the beginning and at the end of the above algorithm. 
\end{proof}

We sketch a proof for the lower bound in the following. 
\begin{lemma}
For any $d$-CQ or $d$-QC schemes $\A$, the probability that $\A$ solves in-place $\dSSP$ is negligible.
\end{lemma}

To prove that in-place $\dSSP$ is hard for $d$-depth quantum circuit, we need to prove the oneway-to-hiding lemma (Lemma 5.7 in~\cite{CCL19}) for in-place shuffling oracle.

\begin{claim}[in-place oracle version of Lemma 5.7 in~\cite{CCL19}]
\label{clm:ip_o2h}
Let $\mathcal{F}$ be a $d$-shuffling of a random Simon's function $f$ and $\mathcal{U}:=\{U_{f_0},\dots,U_{f_d}\}$ be the corresponding in-place $d$-shuffling. Let $\mathbf{S} = \{\bar{S}^{(0)},\dots,\bar{S}^{(d)}\}$ be a sequence of hidden sets as defined in Definition 5.2 in~\cite{CCL19}. Then, for all $k = 0,\dots,d$, there exists a shadow $\mathcal{G}$ of $\mathcal{F}$ in $\bar{S}^{(k)}$ such that for any single-depth quantum circuit $U_c$, initial state $\rho$, and any binary string $t$, 
\begin{align*}
    |\Pr[\Pi_{0/1}\circ \mathcal{U}_{\mathcal{F}}U_c(\rho) = t] -  \Pr[\Pi_{0/1}\circ \mathcal{U}_{\mathcal{G}}U_c(\rho) = t]| &\leq B(\mathcal{U}_{\mathcal{F}}U_c(\rho),\mathcal{U}_{\mathcal{G}}U_c(\rho))\\ 
    &\leq \sqrt{2\Pr[\mbox{find }\bar{S}^{(k)}: U^{\mathcal{F}\setminus\bar{S}^{(k)}},\rho}].  
\end{align*}
Here, $B(\cdot,\cdot)$ is the Bures distance between quantum states, and $U^{\mathcal{F}\setminus\bar{S}^{(k)}}$ is defined in Definition 5.6 in~\cite{CCL19}. $\mathcal{U}_{\mathcal{G}}$ is the in-place oracle for $\mathcal{G}$. 
\end{claim}

\begin{proof}

We first define the shadow $\mathcal{G}$ of $\mathcal{F}$ in $\bar{S}^{(k)}$. The definition follows the same spirit of the shadow in~\cite{CCL19}. In the original definition, the shadow $\mathcal{G}$ maps $x\in \bar{S}^{(k)}$ to a special symbol $\bot$ and is consistent with $\mathcal{F}$ for $x\notin\bar{S}^{(k)}$. This definition of shadow does not work in the in-place oracle setting since the corresponding oracle is not a unitary. 

So, here, we define $\mathcal{G}$ as a random function satisfying the following: If $x\notin \bar{S}^{(k)}$, we let $\mathcal{G}(x) = \mathcal{F}(x)$; else if $x\in \bar{S}^{(k)}$, we let $\mathcal{G}(x)$ to be independent of $\mathcal{F}(x)$, and the in-place oracle of $\mathcal{G}$, $U_\mathcal{G}$, is still a unitary. In particular, for shadows corresponding to $f_1,\dots,f_{d-1}$ in $\bar{S}^{(k)}$, we pick another random permutation that is independent of $f_1,\dots,f_{d-1}$ in $\bar{S}^{(k)}$ and is consistent with $f_1,\dots,f_{d-1}$ for $x\notin \bar{S}^{(k)}$; for $f_d$, we can pick another 2-to-1 function that results in no hidden shift or a different hidden shift. 

Then, the rest of the proof directly follows from the proof for Lemma 5.7 in~\cite{CCL19}.  

\end{proof}

The rest of the analysis to show that in-place $\dSSP$ is hard for $d$-depth quantum circuit in the presence of classical computation follows the proof in~\cite{CCL19} by using the new shadow we construct in the proof for \clm{ip_o2h}. 
 
\end{document}